\documentclass[11pt,letterpaper]{article}

\usepackage{amssymb,amsmath}
\usepackage{amsthm}
\usepackage{lmodern}
\usepackage[T1]{fontenc}
\usepackage{xspace}
\usepackage{color}
\usepackage{newclude}
\usepackage{setspace}
\usepackage[margin=1in]{geometry}

\usepackage{ifpdf}
\ifpdf    
\usepackage[activate={true,nocompatibility},final,tracking=true,kerning=true,spacing=true,factor=1100,stretch=10,shrink=10]{microtype}
\usepackage[usenames,dvipsnames]{xcolor}
\usepackage[colorlinks=true,allcolors=Maroon]{hyperref}
\else    
\usepackage[final]{microtype}
\usepackage[usenames,dvipsnames]{xcolor}
\usepackage[colorlinks=true,allcolors=Maroon]{hyperref}
\fi

\theoremstyle{plain}
\newtheorem{theorem}{Theorem}[section]
\newtheorem{lemma}[theorem]{Lemma}
\newtheorem{fact}[theorem]{Fact}
\newtheorem{corollary}[theorem]{Corollary}

\newtheorem{claim}[theorem]{Claim}
\newtheorem{question}[theorem]{Question}

\theoremstyle{definition}
\newtheorem{definition}[theorem]{Definition}

\newcommand{\zo}{\{0,1\}}
\newcommand{\pmone}{\{\pm 1\}}
\newcommand{\vone}{\vec 1}
\newcommand{\eps}{\varepsilon}
\newcommand{\set}[1]{\left\{#1\right\}}

\DeclareMathOperator*{\EX}{\mathbb E}

\newcommand{\Expp}[2]{\EX_{#1} #2}
\newcommand{\N}{\mathbb N}
\newcommand{\R}{\mathbb R}

\newcommand{\Dc}{\mathcal{D}}

\DeclareMathOperator{\DTEP}{DTEP}
\DeclareMathOperator{\EMD}{{EMD}}
\DeclareMathOperator{\sk}{{\bf sk}}

\providecommand{\card}[1]{\lvert#1\rvert}
\providecommand{\norm}[1]{\lVert#1\rVert}
\providecommand{\tuple}[1]{{\langle{#1}\rangle}}
\newcommand{\indic}[1]{ \mathbf{1}_{\{#1\}} }
\newcommand{\Poincare}{Poincar\'{e}\xspace}
\newcommand{\Patrascu}{P{\v{a}}tra\c{s}cu\xspace}

\newcommand{\Rbb}{\mathbb{R}}

\newcommand{\Kc}{\mathcal{K}}
\newcommand{\Lc}{\mathcal{L}}

\DeclareMathOperator{\supp}{supp}
\DeclareMathOperator{\poly}{poly}

\def\compactify{\itemsep=0pt \topsep=0pt \partopsep=0pt \parsep=0pt}

\newcommand{\aanote}[1]{}
\newcommand{\rnote}[1]{}
\newcommand{\irnote}[1]{}

\newcommand{\myparagraph}[1]{\vspace{1ex}\noindent{\bf #1}}

\title{Sketching and Embedding are Equivalent for Norms%
\thanks{An extended abstract appeared in proceedings of the 47th ACM Symposium on the Theory of Computing (STOC~'2015).
This work was done in part while all authors were at Microsoft
Research Silicon Valley, as well as when the first author was at the
Simons Institute for the Theory of Computing at UC Berkeley.}
}
\author{Alexandr Andoni\footnote{
Columbia University, \texttt{andoni@cs.columbia.edu}}
\and
Robert Krauthgamer\footnote{%
Weizmann Institute of Science, \texttt{robert.krauthgamer@weizmann.ac.il}.
Work supported in part by a US-Israel BSF grant \#2010418
and an Israel Science Foundation grant \#897/13.}
\and
Ilya Razenshteyn\footnote{MIT CSAIL, \texttt{ilyaraz@mit.edu}}
}

\begin{document}

\maketitle

\begin{abstract}
An outstanding open question posed by Guha and Indyk in 2006 
asks to characterize metric spaces in which distances can be estimated
using \emph{efficient sketches}.  Specifically, we say that a~sketching
algorithm is efficient if it achieves constant approximation
using constant sketch size.  A well-known result of Indyk (J.~ACM,
2006) implies that a metric that admits a constant-distortion
embedding into $\ell_p$ for $p\in(0,2]$ also admits an efficient
  sketching scheme.  But is the converse true, i.e., is embedding into
  $\ell_p$ the only way to achieve efficient sketching?

We address these questions for the important special case of \emph{normed spaces},
by providing an almost complete characterization of sketching
in terms of embeddings.
In particular, we prove that a finite-dimensional
normed space allows efficient sketches
if and only if it embeds (linearly) into $\ell_{1-\eps}$ with constant
distortion.
We further prove that for norms that are closed under sum-product,
efficient sketching is equivalent to embedding into $\ell_1$ with constant distortion.
Examples of such norms include the \emph{Earth Mover's Distance}
(specifically its norm variant, called Kantorovich--Rubinstein norm),
and the \emph{trace norm} (a.k.a. Schatten $1$-norm or the nuclear norm).
Using known non-embeddability theorems for these norms by
Naor and Schechtman (SICOMP, 2007) and by Pisier (Compositio.\ Math., 1978),
we then conclude that these spaces do not admit efficient sketches either,
making progress towards answering another open question
posed by Indyk in 2006.

Finally, we observe that resolving whether
``sketching is equivalent to embedding into $\ell_1$ for general norms''
(i.e., without the above restriction)
is \emph{equivalent} to resolving a well-known open problem in Functional Analysis posed by Kwapien in 1969.
\end{abstract}

\onehalfspacing

\section{Introduction}
\label{sec:intro}

One of the most exciting notions in the modern algorithm design
is that of \emph{sketching}, where an~input is summarized into 
a small data structure.
Perhaps the most prominent use of sketching is to estimate 
\emph{distances between points}, one of the workhorses of similarity search. 
For example, some early uses of sketches have been designed
for detecting duplicates and estimating resemblance between documents
\cite{Bro-resemblance, BGMZ, Char}. Another example is Nearest
Neighbor Search, where many algorithms rely heavily on
sketches, under the labels of dimension reduction (like the
Johnson-Lindenstrauss Lemma~\cite{jl-elmhs-84}) 
or
Locality-Sensitive Hashing (see e.g.~\cite{IM,KOR00,AI-CACM}).
Sketches see widespread use in streaming algorithms,
for instance when the input implicitly defines a high-dimensional vector 
(via say frequencies of items in the stream), and a sketch is used to estimate
the vector's $\ell_p$ norm. 
The situation is similar in compressive sensing, where
acquisition of a signal can be viewed as sketching. 
Sketching---especially of distances such as $\ell_p$ norms---%
was even used to achieve improvements for \emph{classical} computational tasks:
see e.g.\ recent progress on 
numerical linear algebra algorithms \cite{Woodruff-sketchBook}, 
or dynamic graph algorithms \cite{agm12-linearGraphs, KKM13-dynamicGraph}.
Since sketching is a crucial primitive that can lead to many
algorithmic advances, it is important to understand its power and limitations.

A primary use of sketches is for {\em distance estimation} between
points in a metric space $(X,d_X)$, such as the Hamming space.  The
basic setup here asks to design a \emph{sketching} function $\sk:X\to \zo^s$,
so that the distance $d_X(x,y)$ can be estimated given only the
sketches $\sk(x),\sk(y)$.  
In the decision version of this problem, the goal is to determine 
whether the inputs $x$ and $y$ are ``close'' or ``far'',
as formalized by the {\em Distance Threshold Estimation Problem}
\cite{SS02}, denoted $\DTEP_r(X,D)$,
where, for a threshold $r>0$ and approximation $D\ge 1$ given as parameters in advance, 
the goal is to decide whether $d_X(x,y)\le r$ or $d_X(x,y)> D r$.
Throughout, it will be convenient to omit $r$ from the subscript.%
\footnote{When $X$ is a normed space it suffices to consider $r=1$
by simply scaling the inputs $x,y$.
}
Efficient sketches $\sk$
almost always need to be randomized, and hence we allow randomization,
requiring (say) 90\% success probability.

The diversity of applications gives rise to a variety of natural
and important metrics $M$ for which we want to solve DTEP: Hamming
space, Euclidean space, other $\ell_p$ norms, the Earth Mover's Distance,
edit distance, and so forth.
Sketches for Hamming and Euclidean distances are now classic and
well-understood \cite{IM,KOR00}.  In particular, both are
``efficiently sketchable'': one can achieve approximation
$D=O(1)$ using sketch size $s=O(1)$ (most importantly,
independent of the dimension of $X$).  Indyk~\cite{I00b} extended
these results to efficient sketches for every $\ell_p$ norm for
$p\in(0,2]$.
In contrast, for $\ell_p$-spaces with $p>2$, 
efficient sketching (constant $D$ and $s$) was proved impossible 
using information-theoretic arguments \cite{SS02,BJKS}.
Extensive subsequent work investigated sketching of other important metric spaces,%
\footnote{Other metric spaces include
edit distance \cite{BJKK04,BES06,OR-edit,AK07} 
and its variants \cite{CPSV,MS00,CMS01,CM07,CK-Ulam,AIK-product},
the Earth Mover's Distance in the plane or in hypercubes \cite{Char,IT,Char-EMD1,KN,AIK,ADIW-EMD},
cascaded norms of matrices \cite{JW-cascaded},
and the trace norm of matrices \cite{lnw14-schatten}.
}
or refined bounds (like a trade-off between $D$ and $s$) 
for ``known'' spaces.%
\footnote{These refinements include 
the Gap-Hamming-Distance problem
\cite{W04,JKS08-ow,BC-GapHam,BCRVdW10,CR12,Sherstov12,Vidick12},
and LSH in $\ell_1$ and $\ell_2$ spaces \cite{MNP,OWZ14}.
}

These efforts provided beautiful results and techniques 
for many specific settings. 
Seeking a~broader perspective, a foundational question 
has emerged \cite[Question \#5]{streaming-open}:
\begin{question} \label{q1}
Characterize metric spaces which admit efficient sketching.
\end{question}
To focus the question, efficient sketching will mean constant $D$ and
$s$ for us.  Since its formulation circa 2006, progress on this
question has been limited.  The only known characterization is by
\cite{GIM-sketch} for distances that are decomposable by coordinates,
i.e., $d_X(x,y)=\sum_{i=1}^{n} \varphi(x_i,y_i)$ for
some~$\varphi$. In particular, they show a number of general
conditions on $\varphi$ which imply an $\Omega(n)$ sketching
complexity for $d_X$.

\subsection{The embedding approach}

To address DTEP in various metric spaces more systematically,
researchers have undertaken the approach of metric embeddings.  A
\emph{metric embedding} of $X$ is a map $f:X\to Y$ into another
metric space $(Y,d_Y)$. The \emph{distortion} of $f$ is the
smallest $D'\ge 1$ for which there exists a scaling factor $t>0$ such
that
\begin{equation*}
  \forall x,y\in X, \qquad
  d_{Y}(f(x),f(y)) \leq t \cdot d_X(x,y) \leq D'\cdot d_{Y}(f(x),f(y)).
\end{equation*}
If the target metric $Y$ admits sketching with parameters $D$ and $s$,
then $X$ admits sketching with parameters $D\cdot D'$ and $s$,
by the simple composition $\sk': x\mapsto \sk(f(x))$.
This approach of ``reducing'' sketching to embedding has been very successful, 
including for variants of the Earth Mover's Distance \cite{Char,IT,Char-EMD1,NS,AIK}, and
for variants of edit distance
\cite{BES06,OR-edit,CK-Ulam,AIK-product,CPSV,MS00,CMS01,CM07}.  The
approach is obviously most useful when $Y$ itself is efficiently
sketchable, which holds for all $Y=\ell_p$, $p\in(0,2]$ \cite{I00b}
(we note that $\ell_p$ for $0 < p < 1$ is not a metric space,
but rather a \emph{quasimetric} space;
the above definitions of embedding and distortion make sense 
even when $Y$ is a quasimetric, 
and we will use this extended definition liberally).
In fact, the embeddings mentioned above are all into $\ell_1$,
except for \cite{AIK-product} which employs a more complicated
target space. We remark that in many cases the distortion $D'$
achieved in the current literature is not constant and depends on 
the ``dimension'' of $X$. 

Extensive research on embeddability into $\ell_1$ has resulted 
in several important distortion lower bounds.
Some address the aforementioned metrics \cite{KN, NS, KR09, AK07}, 
while others deal with metric spaces arising in rather different contexts
such as Functional Analysis \cite{Pisier78,CK10,CKN11-heisenberg}, 
or Approximation Algorithms \cite{LLR,AR98,KV05,KS09}.
Nevertheless, obtaining (optimal) distortion bounds 
for $\ell_1$-embeddability of several metric spaces of interest,
are still well-known open questions \cite{Mat-open}.

Yet sketching is a more general notion, and one may hope to achieve
better approximation by bypassing embeddings into $\ell_1$. 
As mentioned above, some limited success in bypassing an $\ell_1$-embedding 
has been obtained for a variant of edit distance \cite{AIK-product}, 
albeit with a sketch size depending mildly on the dimension of $X$.
Our results disparage these hopes,
at least for the case of \emph{normed spaces}.

\subsection{Our results}

Our main contribution is to show that efficient sketchability of norms
is {\em equivalent} to embeddability into $\ell_{1-\eps}$ with constant distortion. 
Below we only assert the ``sketching $\implies$ embedding'' direction,
as the reverse direction follows from \cite{I00b}, as discussed above.

\begin{theorem} \label{thm:l_1eps}
Let $X$ be a finite-dimensional normed space, and suppose that $0 < \eps < 1/3$.  If
$X$ admits a sketching algorithm for $\DTEP(X,D)$ for 
approximation $D > 1$ with sketch size $s$, then $X$ linearly embeds
into $\ell_{1-\eps}$ with distortion $D'=O(sD / \eps)$.
\end{theorem}

One can ask whether it is possible to improve Theorem~\ref{thm:l_1eps}
by showing that $X$, in fact, embeds into $\ell_1$.  Since many
non-embeddability theorems are proved for $\ell_1$, such a statement
would ``upgrade'' such results to lower bounds for sketches. Indeed, we show
results in this direction too. First of all, the above theorem also
yields the following statement.

\begin{theorem}
  \label{thm:l_1log}
  Under the conditions of Theorem~\ref{thm:l_1eps}, $X$ linearly embeds into $\ell_1$
  with distortion $O(sD \cdot \log (\dim X))$.
\end{theorem}

Ideally, we would like an even stronger statement: efficient
sketchability for norms is equivalent to embeddability into $\ell_1$
with constant distortion (i.e., independent of the dimension of $X$ as
above).  Such a stronger statement in fact requires the resolution of
an open problem posed by Kwapien in 1969 (see \cite{Kalton85, BL00}).
To be precise, Kwapien asks whether every finite-dimensional normed
space $X$ that embeds into $\ell_{1 - \eps}$ for $0 < \eps < 1$ with
distortion $D_0\ge 1$ must also embed into $\ell_1$ with distortion
$D_1$ that depends only on $D_0$ and $\eps$ but not on the dimension
of $X$ (this is a reformulation of the finite-dimensional version of
the original Kwapien's question).  In fact, by
Theorem~\ref{thm:l_1eps}, the ``efficient sketching $\implies$
embedding into $\ell_1$ with constant distortion'' statement is {\em
  equivalent} to a positive resolution of the Kwapien's problem.
Indeed, for the other direction, consider a potential counter-example
to the Kwapien's problem, i.e., a normed space $X$ that embeds
into $\ell_{1-\eps}$ with a constant distortion $D_0\ge1$, but every embedding of
$X$ into $\ell_1$ incurs a distortion $D_1=\omega(1)$,
where the asymptotics is with the dimension of $X$
(it is really a sequence of normed spaces). 
Hence, $X$ admits an efficient sketch obtained by combining 
the embedding into $\ell_{1-\eps}$ with the sketch of~\cite{I00b}, 
but does not embed into $\ell_1$ with constant distortion. 
Thus, if the answer to Kwapien's question is negative,
then our desired stronger statement is false.

To bypass the resolution of the Kwapien's problem, we prove the
following variant of the theorem using a result of
Kalton~\cite{Kalton85}: efficient sketchability is equivalent to
$\ell_1$-embeddability with constant distortion for norms that are
\emph{``closed'' under sum-products}.  A sum-product of two normed
spaces $X$~and~$Y$, denoted $X \oplus_{\ell_1} Y$, is the normed space
$X \times Y$ endowed with $\|(x, y)\| = \|x\| + \|y\|$.  It is easy to
verify that $\ell_1$, the Earth Mover's Distance, and the trace norm
are all closed under taking sum-products (potentially with an increase
in the dimension).  Again, we only need to show the ``sketching
$\implies$ embedding'' direction, as the reverse direction follows
from the arguments above ---
if a normed space $X$ embed into $\ell_1$ with constant distortion, 
we can combine it with the $\ell_1$ sketch of \cite{I00b} 
and obtain an efficient sketch for $X$.  
We discuss the
application of this theorem to the Earth Mover's Distance in
Section~\ref{sec:applications}.

\begin{theorem} \label{thm:l_1}
Let $(X_n)_{n=1}^{\infty}$ be a sequence of finite-dimensional normed
spaces.  Suppose that for every $i_1, i_2 \geq 1$ there exists
$m = m(i_1, i_2) \geq 1$ such that $X_{i_1} \oplus_{\ell_1} X_{i_2}$ embeds isometrically into~$X_m$.
Assume that every $X_n$ admits a sketching algorithm for
$\DTEP(X_n,D)$ for fixed approximation
$D > 1$ with fixed sketch size $s$ (both independent of $n$).
Then, every $X_n$ linearly embeds into $\ell_1$ with bounded distortion (independent of $n$).
\end{theorem}

Overall, we almost completely characterize the norms that are
efficiently sketchable,
thereby making a significant progress on Question~\ref{q1}.  
In particular, our results suggest that the embedding approach (embed into
$\ell_p$ for some $p\in(0,2]$, and use the sketch from \cite{I00b}) 
is essentially unavoidable for norms. 
It is interesting to note that for general metrics (not norms) 
the implication ``efficient sketching$\implies$embedding into $\ell_1$ with constant distortion'' is false: 
for example the Heisenberg group embeds into $\ell_2$-squared 
(with bounded distortion) and hence is efficiently sketchable,
but it is not embeddable into~$\ell_1$~\cite{LN06-heisenberg,CK10,CKN11-heisenberg}
(another example of this sort is provided by Khot and Vishnoi~\cite{KV05}). 
At the same time, we are not aware of any counter-example to
the generalization of Theorem~\ref{thm:l_1eps} to general metrics.

\subsection{Applications}
\label{sec:applications}

To demonstrate the applicability of our results to concrete questions
of interest, we consider two well-known families of normed spaces, for
which we obtain the first non-trivial lower bounds on the sketching
complexity.

\myparagraph{Trace norm.}
Let $\mathcal{T}_n$ be the vector space $\Rbb^{n\times n}$ 
(all real square $n \times n$ matrices) equipped with the trace norm 
(also known as the nuclear norm and Schatten $1$-norm), which is defined to be the sum of singular values.
It is well-known that $\mathcal{T}_n$ embeds into $\ell_2$ (and thus also 
into $\ell_1$) with distortion $\sqrt{n}$ 
(observe that the trace norm is within $\sqrt{n}$ from the Frobenius norm, which
embeds isometrically into $\ell_2$).
Pisier~\cite{Pisier78} proved a matching lower bound of $\Omega(\sqrt{n})$ 
for the distortion of any embedding of~$\mathcal{T}_n$ into $\ell_1$.

This non-embeddability result, combined with our Theorem~\ref{thm:l_1log}, implies a sketching lower bound for the trace norm. 
Before, only lower bounds for specific types of sketches (linear and bilinear) were known~\cite{lnw14-schatten}.

\begin{corollary}
For any sketching algorithm for $\DTEP(\mathcal{T}_n, D)$ with sketch size $s$ the following bound must hold:
$$
sD = \Omega\left(\frac{\sqrt{n}}{\log n}\right).
$$
\end{corollary}

\myparagraph{Earth Mover's Distance.}
The (planar) Earth Mover's Distance (also known as the transportation distance, 
Wasserstein-$1$ distance, and Monge-Kantorovich distance) 
is the vector space $\EMD_n = \{p\in \Rbb^{[n]^2}:\ \sum_i p_i=0\}$ 
endowed with the norm $\|p\|_{\EMD}$ defined as
the minimum cost needed to transport the ``positive part'' of $p$ 
to the ``negative part'' of $p$, 
where the transportation cost per unit between two points in the grid $[n]^2$ is their $\ell_1$-distance 
(for a formal definition see~\cite{NS}).
It is known that this norm embeds into $\ell_1$ 
with distortion $O(\log n)$~\cite{IT,Char,NS},
and that any $\ell_1$-embedding requires distortion $\Omega(\sqrt{\log n})$
\cite{NS}.

We obtain the first sketching lower bound for $\EMD_n$, which in
particular addresses a well-known open question 
\cite[Question \#7]{streaming-open}. 
Its proof is a direct application of Theorem~\ref{thm:l_1} 
(which we \emph{can} apply, since $\EMD_n$ is obviously
closed under taking sum-products), 
to essentially ``upgrade'' the
known non-embeddability into $\ell_1$ \cite{NS} to non-sketchability. 

\begin{corollary}
No sketching algorithm for $\DTEP(\EMD_n, D)$ can achieve approximation $D$ 
and sketch size $s$ that are constant (independent of $n$).
\end{corollary}

The reason we can not apply Theorem~\ref{thm:l_1log} and get a clean quantitative
lower bound for sketches of $\EMD_n$ is the factor $\log(\dim X)$ in 
Theorem~\ref{thm:l_1log}. Indeed, the lower bound on the distortion of an embedding of
$\EMD_n$ into $\ell_1$ proved in~\cite{NS} is $\Omega(\sqrt{\log n})$, which is smaller than
$\log(\dim X) = \Theta(\log n)$.

We note that $\EMD_n$ is a (slight) generalization of the $\EMD$
metric version commonly used in computer science applications. 
In the latter, given two weighted sets $A,B\subset [n]^2$ of the same 
total weight, one has to solve, using only their sketches $\sk(A), \sk(B)$, 
the $\DTEP(\EMD, D)$ problem where the $\EMD$ distance is the min-cost
matching between $A$ and $B$. 
Observe that the weights used in the sets $A,B\subset [n]^2$ are all positive.
The slight difference is that in $DTEP(\EMD_n,D)$, 
which asks analogously to estimate $\|p-q\|_{\EMD}$,
each of $p,q\in \Rbb^{[n]^2}$ has both ``positive'' and ``negative'' parts.
Nevertheless, we show in Appendix~\ref{apx:emdReduction} that
efficient sketching of $\EMD$ on weighted sets implies efficient
sketching of the $\EMD_n$ norm.  Hence, the non-sketchability of
$\EMD_n$ norm applies to $\EMD$ on weighted sets as well.

\subsection{Other related work}

Another direction for ``characterizing tractable metrics'' 
is in the context of streaming algorithms, 
where the input is an implicit vector~$x\in \Rbb^n$ given in the form 
of updates $(i,\delta)$, with the semantics that coordinate $i$ has to
be increased by $\delta\in \Rbb$. 

There are two known results in this vein. First,
\cite{BO-zeroonelaw10} characterized the streaming complexity of
computing the sum $\sum_i \varphi(x_i)$, for some fixed $\varphi$ (e.g.,
$\varphi(x)=x^2$ for $\ell_2$ norm), when the updates are positive. They
gave a precise property of $\varphi$ that determines whether the
complexity of the problem is small. Second, \cite{lnw14-linear} showed
that, in certain settings, streaming algorithms may as well be {\em
  linear}, i.e., maintain a sketch $f(x)=Ax$ for a matrix $A$, and the size of
the sketch is increased by a factor logarithmic in the dimension of $x$.

Furthermore, after the appearance of the conference version of
the current article, there has been another characterization result that
significantly generalizes and extends \cite{BO-zeroonelaw10}. 
Specifically, for every \emph{symmetric} norm $\|\cdot\|_X$, 
it is proved in \cite{BBCKY15} that the sketching (and streaming)
complexity of computing $\|x\|_X$ is characterized by the norm's (maximum)
modulus of concentration, up to polylogarithmic factors in the
dimension of $X$.
Finally, we mention a related work~\cite{ANRW17}, where an efficient
data structure for the Approximate Nearest Neighbor search (ANN) is constructed
for every \emph{symmetric} norm. It is known~\cite{IM, KOR00} that efficient sketches
imply good data structures for ANN, however, the result of~\cite{ANRW17} shows
that having efficient ANN data structure is a way more general property of
an underlying norm.

\subsection{Proof overview}
\label{sec:techniques}

Following common practice, we think of sketching as a communication protocol. 
In fact, our results hold for protocols with an \emph{arbitrary} number of rounds 
(and access to public randomness). 

Our proof of Theorem~\ref{thm:l_1eps} can be divided into two parts:
\emph{information-theoretic} and \emph{analytic}.  First, we use
information-theoretic tools to convert an efficient \emph{protocol} for 
$\DTEP(X, D)$ into a~so-called \emph{threshold map} from $X$ to a Hilbert
space.  Our notion of a threshold map can be viewed as a very
weak definition of embeddability (see Definition~\ref{def_th} for
details).  Second, we use techniques from nonlinear functional
analysis to convert a threshold map to a \emph{linear map} into
$\ell_{1 - \eps}$.

\myparagraph{Information-theoretic part.} To get a threshold map from a protocol for $\DTEP(X, D)$,
we proceed in three steps. First, using the fact that $X$ is a \emph{normed space}, we are able to
give a good protocol for $\DTEP(\ell_{\infty}^k(X), D k)$ (Lemma~\ref{lem:fold}).
The space $\ell_{\infty}^k(X)$ is a product of $k$ copies of $X$ equipped 
with the norm 
$\|(x_1, \ldots, x_k)\| = \max_{i} \|x_i\|$.
Then, invoking the main result from~\cite{AJP-sketch}, we conclude non-existence
of certain Poincar\'{e}-type inequalities for $X$ 
(Theorem~\ref{thm:AJP}, in the contrapositive).

Finally, we use convex duality together with a compactness argument to conclude the existence of a desired
threshold map from $X$ to a Hilbert space 
(Lemma~\ref{lem:Poincare2Thresh}, again in the contrapositive).

\myparagraph{Analytic part.} 
We proceed from a threshold map by upgrading it to a \emph{uniform embedding}
(see Definition~\ref{uniform_def})
of $X$ into a Hilbert space (Theorem~\ref{thm:Threshold2Uniform}).
For this we adapt arguments from~\cite{JR06, R06}. 
We use two tools from nonlinear functional analysis: 
an extension theorem for $1/2$-H\"{o}lder maps from a (general)
metric space to a Hilbert space~\cite{M70} (Theorem~\ref{extension}), 
and a symmetrization lemma for maps
from metric abelian groups to Hilbert spaces~\cite{AMM85} (Lemma~\ref{symmetrization}).

Then we convert a uniform embedding of $X$ into a Hilbert space to a
\emph{linear} embedding into~$\ell_{1 - \eps}$ by applying the result
of Aharoni, Maurey and Mityagin~\cite{AMM85} together with the result
of Nikishin~\cite{Nikishin72}. A similar argument has been used
in~\cite{NS}.

To prove a quantitative version of this step, we examine the proofs
from~\cite{AMM85} and~\cite{Nikishin72}, and obtain
explicit bounds on the distortion of the resulting map. We accomplish
this in Section~\ref{quant_sect}.

\myparagraph{Embeddings into $\ell_1$.}  To prove
Theorem~\ref{thm:l_1log} (which has dependence on the dimension of
$X$), we note that it is a simple corollary of Theorem~\ref{thm:l_1eps} and
a result of Zvavitch~\cite{Z00}, which gives a dimension reduction for
subspaces of $\ell_{1 - \eps}$.

\myparagraph{Norms closed under sum-product.}  Finally, we prove
Theorem~\ref{thm:l_1} --- embeddability into $\ell_1$ for norms closed
under sum-product --- by proving and using a finitary version of the
theorem of Kalton~\cite{Kalton85} (Lemma~\ref{lem:L1Embedding}),
instead of invoking Nikishin's theorem as above. We prove the finitary
version by reducing it to the original statement of Kalton's theorem
via a compactness argument.

Let us point out that Naor and Schechtman~\cite{NS} showed how to use
(the original) Kalton's theorem to upgrade a uniform embedding of
$\EMD_n$ into a Hilbert space to a linear embedding into~$\ell_1$
(they used this reduction to exclude uniform embeddability of
$\EMD_n$). Their proof used certain specifics of $\EMD$. In contrast,
to get Theorem~\ref{thm:l_1} for general norms, we seem to need a~%
finitary version of Kalton's theorem.

We also note that in Theorems~\ref{thm:l_1eps}, \ref{thm:l_1log} 
and~\ref{thm:l_1}, we can conclude
embeddability into $\ell_{1 - \eps}^d$ and $\ell_1^d$ respectively,
where $d$ is \emph{near-linear} in the dimension of the original
space. This conclusion uses the known dimension reduction theorems for
subspaces from~\cite{T90,Z00}.

\section{Preliminaries on functional analysis}

We remind a few definitions and standard facts from functional analysis that
will be useful for our proofs.
A central notion in our proofs is the notion of \emph{uniform
  embeddings}, which is a weaker version of embeddability.

\begin{definition}
    \label{uniform_def}
    For two \emph{metric spaces} $X$ and $Y$, we say that a map $f \colon X \to Y$
    is a \emph{uniform embedding}, if there exist two non-decreasing functions $L, U \colon \Rbb_+ \to \Rbb_+$
    such that for every $x_1, x_2 \in X$ one has
    $L(d_X(x_1, x_2)) \leq d_Y(f(x_1), f(x_2)) \leq U(d_X(x_1, x_2))$,
    $U(t) \to 0$ as $t \to 0$ and $L(t) > 0$ for every $t > 0$.
    The functions $L(\cdot)$ and $U(\cdot)$ are called \emph{moduli} of the embedding.
\end{definition}

\begin{definition}
  An \emph{inner-product space} is a real vector space $X$ together with an \emph{inner product}
  $\langle \cdot, \cdot \rangle \colon X \times X \to \Rbb$, which is a symmetric positive-definite bilinear form.
  A \emph{Hilbert space} is an inner-product space $X$ that is \emph{complete} as a metric space.
\end{definition}

Every inner-product space is a normed space: we can set $\|x\| =
\sqrt{\langle x, x \rangle}$. For a normed space $X$ we denote by $B_X$
its closed unit ball.
The main example of a Hilbert space is $\ell_2$, the space of all real sequences $\{x_n\}$ with $\sum_i x_i^2 < \infty$,
where the inner product is defined as
$$
\langle x, y \rangle = \sum_i x_i y_i.
$$

\begin{definition}
\label{def_kernel}
  For a set $S$, a function $K \colon S \times S \to \Rbb$ is called a \emph{kernel} if $K(s_1, s_2) = K(s_2, s_1)$ for every
  $s_1, s_2 \in S$.
\label{def_pos_neg_kernel}
  We say that the kernel $K$ is \emph{positive-definite} if for every $\alpha_1, \alpha_2, \ldots,
  \alpha_n \in \Rbb$ and $s_1, s_2, \ldots, s_n \in S$, one has
  $$
  \sum_{i,j=1}^n \alpha_i \alpha_j K(s_i, s_j) \geq 0.
  $$
  We say that $K$ is \emph{negative-definite} if for every $\alpha_1, \ldots, \alpha_n \in \Rbb$
  \emph{with $\alpha_1 + \alpha_2 + \ldots + \alpha_n = 0$}
  and $s_1, s_2, \ldots, s_n \in S$, one has
  $$
  \sum_{i,j=1}^n \alpha_i \alpha_j K(s_i, s_j) \leq 0.
  $$
\end{definition}

The following are standard facts about positive- and negative-definite kernels.

\begin{fact}[\cite{S35}]
  \label{char_dot}
  For a kernel $K \colon S \times S \to \Rbb$, there exists an embedding $f \colon S \to H$, where $H$ is a Hilbert space,
  such that $K(s_1, s_2) = \langle f(s_1), f(s_2)\rangle_H$ for every $s_1, s_2 \in S$, iff $K$ is positive-definite.
\end{fact}

\begin{fact}[\cite{Sch38}]
  \label{char_dist}
  For a kernel $K \colon S \times S \to \Rbb$, there exists an embedding $f \colon S \to H$, where $H$ is a Hilbert space,
  such that $K(s_1, s_2) = \|f(s_1) - f(s_2)\|_H^2$ for every $s_1, s_2 \in S$, iff $K(s, s) = 0$ for every $s \in S$ and
  $K$ is negative-definite.
\end{fact}

\begin{definition}
  For an abelian group $G$, we say that a function $f \colon G \to \Rbb$ is \emph{positive-definite} if a kernel $K(g_1, g_2) = f(g_1 - g_2)$
  is positive-definite.
  Similarly, $f$ is said to be \emph{negative-definite} if $K(g_1, g_2) = f(g_1 - g_2)$ is negative-definite.
\end{definition}

The following lemma essentially says that an embedding of an abelian group $G$
into a Hilbert space can be made translation-invariant.

\begin{lemma}[see the proof of Lemma~3.5 in~\cite{AMM85}]
  \label{symmetrization_dot}
  Suppose that $f$ is a map from an abelian group $G$ to a Hilbert space such that
  for every $g \in G$ we have $\sup_{g_1 - g_2 = g} \langle f(g_1), f(g_2)\rangle < +\infty$.
  Then, there exists a map $f'$ from $G$ to a Hilbert space such that $\langle f'(g_1), f'(g_2) \rangle$
  depends only on $g_1 - g_2$ and for every $g_1, g_2 \in G$ we have
  $$
  \inf_{g_1' - g_2' = g_1 - g_2} \langle f(g_1'), f(g_2')\rangle \leq \langle f'(g_1), f'(g_2) \rangle
  \leq \sup_{g_1' - g_2' = g_1 - g_2} \langle f(g_1'), f(g_2') \rangle.
  $$
\end{lemma}

Finally, let $\dim X$ denote the dimension of a finite-dimensional vector space $X$.

\section{Preliminaries on communication complexity}
\label{sec:prelim_cc}

Let $X$ be a metric space, on which we would like to solve $\DTEP_r(X,
D)$ defined as follows for some $r>0$ and $D\ge 1$.  Alice has a point
$x \in X$, Bob has a point $y \in X$, and they would like to decide
between the two cases: $d_X(x, y) \leq r$ and $d_X(x, y) > Dr$. To
accomplish this goal, Alice and Bob exchange at most $s$ bits of
communication. 

There are several types of communication protocols that we consider,
depending on the randomness used, which we present below in the order
of their power. Our main result applies to the most powerful type. We
will later show some connections between the protocols of different
types.

\begin{itemize}
\item \textbf{Deterministic protocols.}  This is a simple two-way
  communication protocol with no randomness.  First, Alice sends a bit
  to Bob that depends only on $x$.  Then, Bob sends a bit to Alice
  that depends on Alice's first communication bit and on $y$.  Then, Alice sends
  a bit to Bob that depends on $x$ and the two previous communication bits, etc.
  Finally, whoever sends the $s$-th bit must decide the answer to the
  $\DTEP$ problem.  We define a \emph{transcript} $\Pi_{x, y}$ to be
  the sequence of $s$ bits sent by the two parties for a given pair of
  inputs $x$ and $y$.

\item \textbf{Private-coin protocols with bounded number of coins.}
This is a randomized version of the previous definition.
Alice and Bob each have access to an \emph{independent} random string,
denoted $a\in \{0, 1\}^R$ and $b\in \{0, 1\}^R$, respectively. 
Communication bits sent by Alice are allowed to depend on $a$, and 
those sent by Bob may depend on $b$.  
We require that for every pair of inputs, the probability (over the
random coins $a,b$) of the answer being correct is at least, say,
$2/3$.  Whenever we allow randomness, the
transcript $\Pi_{x, y}$ becomes a random variable (depending on $x$
and $y$). For fixed $a,b \in \{0, 1\}^R$, we denote by $\Pi_{x, y}(a,b)$
the (deterministic) transcript for inputs $x$ and $y$ when the
random strings are set to be $a$ and $b$, respectively.

\item \textbf{Public-coin protocols with finitely many coins.}
  This is a variant of the previous definition, where Alice and Bob
  have access to a \emph{common} random string sampled uniformly from
  $\{0, 1\}^R$, and the bits sent by both Alice and Bob can depend on
  this random string.  Again, we require the probability (over the
  public coins) of the answer being correct to be at least, say, $2/3$
  for every pair of inputs. Also, we denote by $\Pi_{x, y}(u)$ the
  deterministic transcript for fixed inputs and public coins $u$.

Clearly, a public-coin protocol with $2R$ public coins can emulate a
private-coin protocol with $R$ random bits for each of Alice and Bob.

\item \textbf{Public-coin protocols with countably many coins.}  The
  protocols defined above are standard in
  the communication complexity literature.  However, we need a
  definition that is \emph{stronger:} we allow \emph{countably many
    public coins.} The reason to consider the stronger notion is that
  the known protocols for $\DTEP$ based on~\cite{I00b} fall into this
  category.  Since we allow infinitely many coins, we need to be
  careful when defining a class of allowed protocols.  A sequence of
  coin tosses $u$ can be identified with a point in the Cantor space
  $\Omega = \{0, 1\}^\omega$ equipped with the standard Lebesgue
  measure. We require that for every pair of inputs $x, y \in X$, the
  function $u \mapsto \Pi_{x,y}(u)$ is measurable. This restriction
  allows us to consider probabilities of the form
  $\mathrm{Pr}[\Pi_{x,y} \in A]$, where $A \subseteq \{0, 1\}^s$ is an
  arbitrary set of possible transcripts.  In particular, 
  the probability of success is well-defined, and we require it, 
  as before, to be at least $2/3$.

\end{itemize}

The results in this paper apply to the most general protocols:
public-coin protocols with countably many coins.

We now show some connections between these notions.  A crucial tool in
our result is a theorem of \cite{AJP-sketch}, which is itself based on
the tools from \cite{BJKS-infoTheory}. The latter shows a lower bound
for private-coin protocols with finitely many coins. We 
show next how the lower bounds from \cite{AJP-sketch, BJKS-infoTheory}
extend to the most general type, public-coin protocols with countably
many coins.

\subsection{Information complexity: private-coins vs public-coins}

In general, a lower bound for private-coins protocols does not
imply a lower bound for public-coins protocols (without a loss in
the parameters). However, such an implication does hold for the particular
lower bound technique that we are employing. In particular, we use and
exploit the notion of information complexity
from~\cite{BJKS-infoTheory}, defined as follows. Let $(x, y, \lambda)$
be distributed according to a distribution $\Dc$ over $X \times X
\times \Lambda$, where $\Lambda$ is an auxiliary set. We will assume
that the support of $\Dc$ is \emph{finite}.  Then, we can define the
information complexity with respect to $\mathcal{D}$, denoted
$\mathrm{IC}_{\mathcal{D}}(\DTEP_r(X, D))$, to be the infimum of
$I(x, y : \Pi_{x, y} \mid \lambda)$ over all private-coin protocols
for $\DTEP_r(X, D)$, which succeed on every valid input with
probability at least $2/3$, where $I(\cdot : \cdot \mid \cdot)$ is the
(conditional) mutual information.

It is a standard fact that $\mathrm{IC}_{\mathcal{D}}(\DTEP_r(X, D))$
is a lower bound on the communication complexity of $\DTEP_r(X, D)$
with \emph{private-coin protocols} since
$$
I(x, y : \Pi_{x,y} \mid \lambda) \leq \sup_{x,y,a,b} |\Pi_{x,y}(a, b)|.
$$

However, we are interested in using the information complexity (as defined above) to lower bound
the communication complexity of $\DTEP_r(X, D)$ for \emph{public-coin protocols with finite number of coins}.
It turns out that $\mathrm{IC}_{\mathcal{D}}(\DTEP_r(X, D))$ is
a valid lower bound for this case as well, as argued in the claim below.

\begin{lemma}
\label{public_vs_private}
The communication complexity of $\DTEP_r(X, D)$ for public-coin protocols with finite number of coins
is at least $\mathrm{IC}_{\mathcal{D}}(\DTEP_r(X, D))$.
\end{lemma}
\begin{proof}
Consider any protocol with public randomness, denoted $\Pi_{x,y}(u)$,
where $x,y$ are the two inputs and $u$ is the public random string. 
Then 
$$ 
\sup_{x,y,u}\card{\Pi_{x, y}(u))} \ge H(\Pi_{x,y}(u) \mid \lambda, u) \ge I(x,y : \Pi_{x,y}(u) \mid \lambda, u).
$$

Now consider the following private-coins protocol $\Pi'_{x,y}(
a, b)$, where $a,b$ are the two private random strings of Alice and
Bob, respectively.  In the first round, Alice sends $a$ to Bob to be
used as public randomness $u=a$.  Then they run $\Pi_{x,y}(u)$. In other
words, the transcript of $\Pi'_{x,y}(a,b)$ is $\langle a,
\Pi_{x,y}(a)\rangle$.
We claim that 
$$I(x,y : \Pi'_{x,y}(a,b) \mid \lambda) = I(x,y : \Pi_{x,y}(u) \mid \lambda, u).$$
Indeed, by definition of $\Pi'$,
$$I(x,y : \Pi'_{x,y}(a,b) \mid \lambda) = I(x,y : a, \Pi_{x,y}(a) \mid \lambda),$$
and using the chain rule for mutual information,
$$I(x,y : a, \Pi_{x,y}(a) \mid \lambda) = 
  I(x,y : a \mid \lambda) +
  I(x,y : \Pi_{x,y}(a) \mid \lambda, a).
$$ 
The first term is exactly zero since $x,y$ and $a$ are independent
  (conditioned on $\lambda$).  The remaining term gives the equality
we are looking for, and proves the lemma. 
In particular, we see that the length of a public-coin
protocol is at least the information complexity $I(x,y :
\Pi'_{x,y}(a,b) \mid \lambda)$ of any private-coin protocol $\Pi'$.
\end{proof}

\subsection{From countable to finite number of coins}

We now observe that if we focus only on a {\em finite} number of
possible inputs to our $\DTEP$ problem, then the existence of a
protocol with countably-many coins implies the existence of a protocol
with bounded number of coins. This claim will be sufficient to
generalize our theorem to the most general type of
protocols---public-coin protocols with countably many coins:
see the remark after Theorem~\ref{thm:AJP}.

\begin{claim}
\label{infinite_to_finite}
Fix a public-coin protocol with \emph{countably many coins} and $s$
bits of communication.  Let $(x^{(1)}, y^{(1)}), (x^{(2)},
y^{(2)}), \ldots, (x^{(N)}, y^{(N)})$ be $N$ fixed pairs of
inputs for the $\DTEP$ problem, and let $\eps > 0$ be a positive
parameter.  
Then there exists a public-coin protocol 
\emph{with $R = R(s, N, \eps)<\infty$ coins} and $s$ bits of communication, 
such that for every $1 \leq i \leq N$, the success probabilities of
the original and the new protocols on $(x^{(i)}, y^{(i)})$ 
differ by at most $\eps$.
\end{claim}
\begin{proof}
Since we care about the correctness of the protocol only on the inputs
$(x^{(i)}, y^{(i)})$, we can think of the protocol as a distribution
over a \emph{bounded} (as a function of $s$ and $N$) number of
deterministic protocols (there's only a finite number of distinct
protocol transcripts).  Then, we can approximate this distribution
within a statistical distance $\eps$ using a bounded number of public
coins.
\end{proof}


\section{From sketches to uniform embeddings}

Our main technical result shows that, for a finite-dimensional normed
space $X$, good sketches for $\DTEP(X, D)$ imply a good uniform
embedding of $X$ into a Hilbert space (Definition~\ref{uniform_def}). Below
is the formal statement.

    \begin{theorem}
        \label{main_quant}
        Suppose a finite-dimensional normed space $X$ admits a public-coin randomized communication
        protocol
        for $\DTEP(X, D)$ of size $s$ for approximation $D > 1$.
        Then, there exists a map $f \colon X \to H$ to a Hilbert
        space 
        such that for all $x_1, x_2 \in X$,
        $$
            \min\set{1, \frac{\|x_1 - x_2\|_X}{s \cdot D}} \leq \|f(x_1) - f(x_2)\|_H
            \leq K \cdot \|x_1 - x_2\|_X^{1/2},
        $$
        where $K > 1$ is an absolute constant.
    \end{theorem}

Theorem~\ref{main_quant} implies a \emph{qualitative} version of
Theorem~\ref{thm:l_1eps} using the results of Aharoni, Maurey, and
Mityagin~\cite{AMM85} and Nikishin~\cite{Nikishin72} (see Theorem
\ref{thm:amm_nik}).  

    \begin{theorem}[\cite{AMM85, Nikishin72}] \label{thm:amm_nik}
      For every fixed $0 < \eps < 1$, any finite-dimensional normed space $X$ that is uniformly embeddable into a Hilbert space 
      is linearly embeddable into $\ell_{1 - \eps}$ with a
      distortion that depends only on $\eps$ and the moduli of the assumed uniform embedding.
    \end{theorem}

To prove the full (quantitative) versions of
Theorems~\ref{thm:l_1eps} and~\ref{thm:l_1log}, we adapt the proofs
from~\cite{AMM85, Nikishin72} in Section~\ref{quant_sect} to get an explicit bound
on the distortion.

In the rest of this section, we prove Theorem~\ref{main_quant} according to 
the outline in Section~\ref{sec:techniques},
putting the pieces together in Section~\ref{sec:together}.

\subsection{Sketching implies the absence of Poincar\'{e} inequalities}
\label{sec:AJP}

Sketching is often viewed from the perspective of 
a two-party communication complexity.
Alice receives input $x$, Bob receives $y$, and they need to communicate 
to solve the $\DTEP$ problem. 
In particular, a sketch of size $s$ implies a communication protocol 
that transmits $s$ bits: 
Alice just sends her sketch $\sk(x)$ to Bob, who computes the 
output of $\DTEP$ (based on that message and his sketch $\sk(y)$). 
We assume here a public-coins model, i.e., 
Alice and Bob have access to a common (public) random string 
that determines the sketch function $\sk$.

To characterize sketching protocols, we build on results of 
Andoni, Jayram and \Patrascu~\cite[Sections 3 and 4]{AJP-sketch}.
This works in two steps: 
first, we show that a protocol for $\DTEP(X, D)$ 
implies a sketching algorithm for $\DTEP(\ell_\infty^k(X), kD)$, 
with a loss of factor $k$ in approximation (Lemma~\ref{lem:fold}, see the proof in the end of the section).
As usual, $\ell_\infty^k(X)$ is a normed space derived from $X$
by taking the vector space $X^k$ and letting the norm of a vector 
$(x_1,\ldots x_k)\in X^k$ be the maximum of the norms of its $k$ components.
The second step is to apply a result from \cite{AJP-sketch} (Theorem~\ref{thm:AJP}), 
which asserts that sketching for $\ell_\infty^k(X)$ 
precludes certain \Poincare inequalities for the space $X$. 

\begin{lemma} \label{lem:fold} 
Let $X$ be a finite-dimensional normed space that for some $D \geq 1$ admits a communication protocol
for $\DTEP(X, D)$ of size $s$.
Then for every integer $k$, the space $\ell_\infty^k(X)$ admits sketching 
with approximation $kD$ and sketch size $s'=O(s)$.
\end{lemma}
\begin{proof}
Fix a threshold $t>0$, 
and recall that we defined the success probability of sketching to be $0.9$.
By our assumption, there is a sketching function $\sk$ for $X$ 
that achieves approximation $D$ and sketch size $s$ for threshold $k t$.
Now define a ``sketching'' function $\sk'$ for $\ell_\infty^k(X)$ 
by choosing random signs $\eps_1,\ldots,\eps_k\in \pmone$,
letting $\sk': x\mapsto \sk(\sum_{i=1}^k \eps_i x_i)$,
and using the same decision procedure used by $\sk$ (for $X$).

Now to examine the performance of $\sk'$, consider $x,y\in
\ell_\infty^k(X)$.  If their distance is at most $t$, then we always have 
that $\norm{\sum_{i=1}^k \eps_i x_i
  - \sum_{i=1}^k \eps_i y_i} \leq \sum_{i=1}^k \norm{x_i-y_i} \leq k
t$ (i.e., for
every realization of the random signs).
Thus with probability at least $0.9$ the sketch will declare that
$x,y$ are ``close''.

If the distance between $x,y$ is greater than $kD\cdot t$,
then for some coordinate, say $i=1$, we have $\norm{x_1-y_1} > kD\cdot t$.
Letting $z=\sum_{i\ge 2} \eps_i(x_i-y_i)$, 
we can write 
$\norm{\sum_{i=1}^k \eps_i x_i - \sum_{i=1}^k \eps_i y_i} 
 = \norm{\eps_1(x_1-y_1) + z}
 = \norm{(x_1-y_1) + \eps_1 z}$.
The last term must be at least $\norm{x_1-y_1}$
under at least one of the two possible realizations of $\eps_1$,
because by the triangle inequality 
$
  2\norm{x_1-y_1} 
  \leq \norm{(x_1-y_1) + z} + \norm{(x_1-y_1) - z} 
$.
We see that with probability $1/2$ we have
$\norm{\sum_{i=1}^k \eps_i x_i - \sum_{i=1}^k \eps_i y_i} 
  \ge \norm{x_1-y_1} 
  > D\cdot k t$,
and thus with probability at least $1/2\cdot 0.9=0.45$ 
the sketch will declare that $x,y$ are ``far''.
This last guarantee is not sufficient for $\sk'$ to be called a sketch,
but it can easily be amplified.

The final sketch $\sk''$ for $\ell_\infty^k(X)$ is obtained 
by $O(1)$ independent repetitions of $\sk'$,
and returning ``far'' if at least $0.3$-fraction of the repetitions 
come up with this decision.
These repetitions amplify the success probability to $0.9$,
while increasing the sketch size to $O(s)$. 
\end{proof}

We now state a slight modification of the theorem from \cite{AJP-sketch}. 
We will actually use its contrapositive,
to conclude the absence of \Poincare inequalities.

\begin{theorem}[modification of \cite{AJP-sketch}]
\label{thm:AJP}
Let $X$ be a metric space, and fix $r>0$, $D\ge 1$.
Suppose there are $\alpha>0$, $\beta\ge 0$, and
two symmetric probability measures $\mu_1, \mu_2$ on $X \times X$ such that
\begin{itemize} \compactify
\item
  The support of $\mu_1$ is finite and is only on pairs with distance at most $r$;
\item
  The support of $\mu_2$ is finite and is only on pairs with distance greater than $Dr$; and
\item
  For every $f: X \to B_{\ell_2}$ (where $B_{\ell_2}$ is the unit ball of $\ell_2$),
  $$
    \Expp{(x, y) \sim \mu_1}{\|f(x) - f(y)\|^2}
    \geq \alpha \cdot
    \Expp{(x, y) \sim \mu_2}{\|f(x) - f(y)\|^2}
    \; - \beta.
  $$
\end{itemize}

Then for every integer $k$, the communication complexity of
$\DTEP(\ell_\infty^k(X), D)$ for protocols with \emph{countably many public coins} (see Section~\ref{sec:prelim_cc} for precise definitions)
and with probability of error $\delta_0>0$ is
at least $\Omega(k)\cdot
\left(\alpha(1-2\sqrt{\delta_0})-\beta\right)$.
\end{theorem}

In~\cite{AJP-sketch}, almost the same theorem is proved with only one
difference: the protocols for $\DTEP(\ell_\infty^k(X), D)$ are only
allowed to use finitely many private coins. Here we use
Claims~\ref{public_vs_private} and~\ref{infinite_to_finite} to
generalize their theorem to Theorem~\ref{thm:AJP}.

Indeed, because the ``hard distributions'' $\mu_1$ and $\mu_2$ are
finitely-supported, an inspection of the proofs from~\cite{AJP-sketch}
shows that there is a finite set of inputs $\mathcal{I}$ such that any
private-coin protocol for $\DTEP(\ell_\infty^k(X), D)$ that is correct
on $\mathcal{I}$ with probability at least $1 - \delta_0$ must have
information complexity at least $\Omega(k)\cdot
\left(\alpha(1-2\sqrt{\delta_0})-\beta\right)$. But by
Claim~\ref{public_vs_private}, we get that any protocol with bounded
number of public coins correct on $\mathcal{I}$ must have
communication complexity at least $\Omega(k)\cdot
\left(\alpha(1-2\sqrt{\delta_0})-\beta\right)$. Finally,
Claim~\ref{infinite_to_finite} implies that the same is true for
protocols with countably many public coins that are correct on all
valid inputs with probability at least $1 - \delta_0$.

\subsection{The absence of Poincar\'{e} inequalities implies threshold maps}
\label{sec:Poincare2Thresh}

We proceed to prove that non-existence of \Poincare inequalities implies the
existence a ``threshold map'', as formalized in
Lemma~\ref{lem:Poincare2Thresh} below.
The proof is similar to duality arguments that one often encounters in embedding theory:
for instance, see Proposition~15.5.2 in~\cite{Matousek-book}.
First we define the notion of
threshold maps.

\begin{definition}
    \label{def_th}
        A map $f \colon X \to Y$ between metric spaces $(X, d_X)$ and $(Y, d_Y)$
        is called an \emph{ $(s_1, s_2, \tau_1, \tau_2, \tau_3)$-threshold map} for $0 < s_1 < s_2$,
        $0 < \tau_1 < \tau_2 < \tau_3$, if
        for all $x_1, x_2 \in X$:
        \begin{itemize} \compactify
            \item if $d_X(x_1, x_2) \leq s_1$, then $d_Y(f(x_1), f(x_2)) \leq \tau_1$;
            \item if $d_X(x_1, x_2) \geq s_2$, then $d_Y(f(x_1), f(x_2)) \geq \tau_2$; and
            \item $d_Y(f(x_1), f(x_2)) \leq \tau_3$.
        \end{itemize}
    \end{definition}

We now provide the main lemma of this section, stated in the contrapositive:
the non-existence of threshold maps implies a \Poincare inequality.

    \begin{lemma} \label{lem:Poincare2Thresh}
        Suppose $X$ is a metric space that does not allow an $(s_1,
        s_2, \tau_1, \tau_2, +\infty)$-threshold map to a Hilbert
        space.  Then, for every $\delta > 0$ there exist two symmetric
        probability measures $\mu_1, \mu_2$ on $X \times X$ such that
        \begin{itemize} \compactify
        \item
          The support of $\mu_1$ is finite and is only on pairs with distance at most $s_1$;
        \item 
          The support of $\mu_2$ is finite and is only on pairs with distance at least $s_2$; and
        \item
          For every $f \colon X \to B_{\ell_2}$,
                \begin{equation}\label{eqn:poincare}
                    \Expp{(x, y) \sim \mu_1}{\|f(x) - f(y)\|^2} 
                    \geq \left( \frac{\tau_1}{\tau_2} \right)^2 \cdot
                    \Expp{(x, y) \sim \mu_2}{\|f(x) - f(y)\|^2}
                    \; - \delta.
                \end{equation}
        \end{itemize}
    \end{lemma}

We prove Lemma~\ref{lem:Poincare2Thresh} via the following three
claims.  The first one uses standard arguments about embeddability of
finite subsets (see, e.g., Proposition 8.12 in~\cite{BL00}, or Lemma 1.1 from~\cite{B92}). We note
that this claim requires a finite value for $\tau_3$, as opposed to
$\tau_3=+\infty$, which is the only reason the definition of a threshold embedding
(Definition~\ref{def_th}) needs the parameter $\tau_3$.
In the following claims, we denote by $\binom{X}{2}$ the set of
all \emph{unordered} pairs $\set{x, y}$ with $x, y \in X$, $x \ne y$.

    \begin{claim}
        \label{compactness}
        For every metric space $X$ and every $0 < s_1 < s_2$, $0 < \tau_1 < \tau_2 < \tau_3$
        there exists an $(s_1, s_2, \tau_1, \tau_2, \tau_3)$-threshold map of $X$
        to a Hilbert space iff the same is true for every finite subset of~$X$.
    \end{claim}
    The proof of Claim~\ref{compactness} uses standard definitions and facts from general topology:
    product topology, Tychonoff's theorem, as well as convergence and accumulation points along nets.
    These definitions can be found in a general topology textbook
    (see, e.g.,~\cite{M00}).
    \begin{proof}
        The ``only if'' direction is obvious, 
        so let us turn to the ``if'' part.
        Consider the topological space
        $$
            U = \prod_{\set{x, y} \in \binom{X}{2}} [-\tau_3^2, \tau_3^2].
        $$
        By Tychonoff's theorem $U$ is compact.
        For every finite $X' \subset X$ there exists an $(s_1, s_2, \tau_1, \tau_2, \tau_3)$-threshold
        map $f_{X'}$ from $X'$ to a Hilbert space.
        It gives rise to a point $u_{X'} \in U$ whose coordinates are given by
        $$
            (u_{X'})_{x,y} = \begin{cases}
                \|f_{X'}(x) - f_{X'}(y)\|^2, & \mbox{if $x, y \in X'$;}\\
                0, & \mbox{otherwise.}
            \end{cases}
        $$
        Since $U$ is compact, $u_{X'}$ has an accumulation point $u^* \in U$ along the net of finite subsets
        of $X$.
        Let us reformulate what it means.
        \begin{claim}
          \label{compactness_1}
          For every $\{x_1, y_1\}, \{x_2, y_2\}, \ldots, \{x_k, y_k\} \in \binom{X}{2}$
          and every $\eps > 0$,
          there exists a finite set $A \subset X$ such that
          for all $1 \leq i \leq k$, both $x_i, y_i\in A$ 
          and $\bigl|(u^*)_{x_i, y_i} - \|f_A(x_i) - f_A(y_i)\|^2\bigr| < \eps$ .
        \end{claim}
        Now we define a kernel $K \colon X \times X \to \Rbb$, given by (recall Definition~\ref{def_kernel}):
$$
K(x,y) = \begin{cases}
0, & \mbox{if $x = y$;}\\
(u^*)_{x,y}, & \mbox{otherwise.}
\end{cases}
$$
\begin{claim}
\label{neg_def_compact}
The kernel $K(\cdot, \cdot)$ is negative-definite.
\end{claim}
\begin{proof}
Suppose that $K$ is \emph{not} negative-definite. It means that there exist
$\alpha_1, \alpha_2, \ldots, \alpha_n \in \Rbb$ with $\sum_i \alpha_i = 0$,
and $t_1, t_2, \ldots, t_n \in X$ such that
$$
\sum_{i,j=1}^n \alpha_i \alpha_j K(t_i, t_j) = \gamma > 0.
$$
There exists $\eps > 0$ such that for every $(a_{ij})_{i,j=1}^n$ with $|a_{ij} - K(t_i, t_j)| < \eps$
one has
\begin{equation}
\label{robust_ineq}
\sum_{i,j=1}^n \alpha_i \alpha_j a_{ij} \geq \gamma/2 > 0.
\end{equation}
Now apply Claim~\ref{compactness_1} to get a finite set $A \subset X$ that contains all $s_i$'s such that
$\|f_A(t_i) - f_A(t_j)\|^2$ is within $\eps$ from $K(t_i, t_j)$ for every $i, j$. But by~(\ref{robust_ineq}), it means that
$$
\sum_{i,j=1}^n \alpha_i \alpha_j \|f_A(t_i) - f_A(t_j)\|^2 \geq \gamma/2 > 0,
$$
which contradicts Fact~\ref{char_dist}.
This proves Claim~\ref{neg_def_compact}.
\end{proof}
Thus, by Fact~\ref{char_dist}, there exists a map $f \colon X \to H$ to a Hilbert space $H$ such that
for every $x, y \in X$ one has $\|f(x) - f(y)\|^2 = K(x, y)$.
The final step is to verify that $f$ is indeed a required $(s_1, s_2, \tau_1, \tau_2, \tau_3)$-map
(according to Definition~\ref{def_th}). This can be done exactly the same way as in the proof of Claim~\ref{neg_def_compact}.
This completes the proof of Claim~\ref{compactness}.
    \end{proof}

    \begin{claim} 
        \label{sdp_duality}
        Suppose that $(X, d_X)$ is a \emph{finite} metric space and $0 < s_1 < s_2$, $0 < \tau_1 < \tau_2 < \tau_3$.
        Assume that there is \emph{no} $(s_1, s_2, \tau_1, \tau_2, \tau_3)$-threshold map
        of $X$ to $\ell_2$.
        Then, there exist two symmetric probability measures $\mu_1, \mu_2$ on $X \times X$ such that
        \begin{itemize}
            \item
                $\mu_1$ is supported only on pairs with distance at most $s_1$, while $\mu_2$ is
                supported only on pairs with distance at least $s_2$; and
            \item
                for every $f \colon X \to \ell_2$, 
                \begin{equation}\label{eqn:poincareTaus}
                    \Expp{(x, y) \sim \mu_1}{\norm{ f(x) - f(y) }^2} \geq \left( \frac{\tau_1}{\tau_2} \right)^2 \cdot
                    \Expp{(x, y) \sim \mu_2}{ \norm{ f(x) - f(y) }^2}
                    - \left( \frac{2 \tau_1}{\tau_3} \right)^2 \cdot \sup_{x \in X} \norm{ f(x) }^2.
                \end{equation}
        \end{itemize}
    \end{claim}
    \begin{proof}
        Let $\Lc_2 \subset \Rbb^{\binom{X}{2}}$ be the cone of squared Euclidean metrics (also known as negative-type distances) on $X$.
        Let $\Kc \subset \Rbb^{\binom{X}{2}}$
        be the polytope of \emph{non-negative} functions $l \colon \binom{X}{2} \to \Rbb_+$
        such that for every $x, y \in X$ we have
        \begin{itemize} \compactify
            \item
                $l(\set{x, y}) \leq \tau_3^2$;
            \item
                if $d_X(x, y) \leq s_1$, then $l(\set{x, y}) \leq \tau_1^2$;
            \item
                if $d_X(x, y) \geq s_2$, then $l(\set{x, y}) \geq \tau_2^2$.
        \end{itemize}

        Notice that $\Lc_2 \cap \Kc = \emptyset$, as otherwise $X$ allows 
        an $(s_1, s_2, \tau_1, \tau_2, \tau_3)$-threshold map to $\ell_2$.
        We will need the following claim, which is just a variant of the 
        Hyperplane Separation Theorem.

        \begin{claim}
            There exists $a \in \Rbb^{\binom{X}{2}}$ such that 
            \begin{align}
              \forall l\in \Lc_2, \qquad & \tuple{a,l} \leq 0; \\
              \forall l\in \Kc, \qquad & \tuple{a,l} > 0.
            \end{align}
        \end{claim}
        \begin{proof}
            Since both $\Lc_2$ and $\Kc$ are convex and closed, and, in addition, $\Kc$ is compact,
            there exists a separating (affine) hyperplane between $\Lc_2$ and $\Kc$.
            Specifically, there is a non-zero $a$ such that 
            for every $l \in \Lc_2$ one has $\langle a, l \rangle \leq \eta$, 
            and for every $l \in \Kc$ one has $\langle a, l \rangle > \eta$.
            Since $\Lc_2$ is a cone, one can assume without loss of generality that $\eta = 0$.
            Indeed, the case $\eta < 0$ is impossible because $0 \in \Lc_2$,
            so suppose that $\eta > 0$. If for all $l \in \Lc_2$ we have $\langle a, l \rangle \leq 0$,
            then we are done. Otherwise, take any $l \in \Lc_2$ such that $\langle a, l \rangle > 0$, and scale it by sufficiently large $C>0$ to get
            a point $Cl\in \Lc_2$ so that $\tuple{a,Cl} = C\tuple{a,l}>\eta$,
            arriving to a contradiction.
        \end{proof}
        
        We now continue the proof of Claim~\ref{sdp_duality}.
        We may assume without loss of generality that 
        \begin{align} \label{eq:a_neg} \textstyle
        \forall \set{x, y} \in \binom{X}{2}, \quad 
        \text{ if } d_X(x, y) < s_2
        \text{ then } a_{\set{x,y}} \le 0. 
        \end{align}
To see this, let us zero every such $a_{\set{x,y}}>0$, and denote the resulting point $\hat a$.
Then for every $l\in\Lc_2$ (which clearly has non-negative coordinates), 
$\tuple{\hat a,l}\leq \tuple{a,l} \leq 0$. 
And for every $l\in\Kc$, let $\hat l$ be equal to $l$ except that 
we zero the same coordinates where we zero $a$ 
(which in particular satisfy $d_X(x, y) < s_2$);
observe that also $\hat l\in\Kc$, and thus 
$\tuple{\hat a,l} = \tuple{a,\hat l} >0$.
We get that $\hat a$ separates $\Kc$ and $\Lc_2$
and also satisfies \eqref{eq:a_neg}.

        Now we define non-negative functions $\widetilde{\mu}_1, \widetilde{\mu}_2, \widetilde{\mu}_3
        \colon \binom{X}{2} \to \Rbb_+$ as follows:
        \begin{eqnarray*}
            \widetilde{\mu}_1(\set{x,y}) 
            & = & -a(\set{x,y})\ \indic{d_X(x, y) \leq s_1} ;
            \\
            \widetilde{\mu}_2(\set{x,y})
            & = & a(\set{x,y})\ \indic{d_X(x, y) \geq s_2 \text{ and } a_{x,y} \geq 0} ;
            \\
            \widetilde{\mu}_3(\set{x,y})   
            & = & -a(\set{x,y})\ \indic{d_X(x, y) > s_1 \text{ and } a_{x,y} < 0} .
        \end{eqnarray*}
        By \eqref{eq:a_neg}, these $\widetilde{\mu}_i$ ``cover'' all cases,
        i.e., 
\begin{align*} \textstyle
  \forall\set{x, y} \in \binom{X}{2}, \quad 
  a(\set{x,y}) = - \widetilde{\mu}_1(\set{x,y}) + \widetilde{\mu}_2(\set{x,y}) - \widetilde{\mu}_3(\set{x,y}).
\end{align*}
For $i \in \set{1,2,3}$ define $\lambda_i = \sum_{\set{x,y}} \widetilde{\mu}_i(\set{x,y})$
        and $\mu_i(\set{x,y}) = \widetilde{\mu}_i(\set{x,y}) / \lambda_i$.
        We argue that $\mu_1$ and $\mu_2$ are as required by Claim~\ref{sdp_duality}, and indeed the only non-trivial property to check is the second item.
        From the condition that $\langle a, l \rangle \leq 0$ for every $l \in \Lc_2$
        we get that for every map $f \colon X \to \ell_2$,
        \begin{multline*}
          0 
          \geq \sum_{\set{x,y}} a(\set{x,y}) \cdot \norm{ f(x) - f(y) }^2
          \\ = \sum_{\set{x,y}} \Big[ - \widetilde{\mu}_1(\set{x,y}) + \widetilde{\mu}_2(\set{x,y}) - \widetilde{\mu}_3(\set{x,y}) \Big]\cdot \norm{ f(x) - f(y) }^2 ,
        \end{multline*}
        which, in turn, implies
        \begin{equation}
            \label{raw_poincare}
            \lambda_1 \cdot \Expp{(x, y) \sim \mu_1}{\|f(x) - f(y)\|^2}
            \geq
            \lambda_2 \cdot \Expp{(x, y) \sim \mu_2}{\|f(x) - f(y)\|^2}
            - 4 \lambda_3 \cdot \sup_x \|f(x)\|^2 .
        \end{equation}
        Consider the point $l \in \Kc$ with value $\tau_1^2$ on $\supp(\mu_1)$,
        value $\tau_2^2$ on $\supp(\mu_2)$, value $\tau_3^2$ on $\supp(\mu_3)$,
        and $0$ otherwise; 
        the condition $\langle a, l \rangle > 0$ gives 
        $$
            - \lambda_1 \tau_1^2 + \lambda_2 \tau_2^2 - \lambda_3 \tau_3^2 > 0, 
        $$
        which implies $\lambda_1 < \lambda_2 \cdot \tau_2^2 / \tau_1^2$
        and $\lambda_3 < \lambda_2 \cdot \tau_2^2 / \tau_3^2$ (in particular, $\lambda_2 > 0$).
        Plugging into~\eqref{raw_poincare}, we get the inequality required for
        Claim~\ref{sdp_duality}.
    \end{proof}

We are now ready to prove Lemma~\ref{lem:Poincare2Thresh}.

    \begin{proof}[Proof of Lemma~\ref{lem:Poincare2Thresh}]
      We start with a metric space $X$ that does not admit a
      $(s_1,s_2,\tau_1,\tau_2,+\infty)$-threshold map, and prove that
      this implies the \Poincare inequality~\eqref{eqn:poincare}.

      Indeed, $X$ has no $(s_1, s_2, \tau_1, \tau_2,
      \tau_3)$-threshold map to a Hilbert space for any finite value
      $\tau_3$. We set $\tau_3 > \tau_2$ be sufficiently large so that
      $(2 \tau_1/\tau_3)^2 < \delta$.  Then, by
      Claim~\ref{compactness} there exists a finite subset $X' \subset
      X$ that has no $(s_1, s_2, \tau_1, \tau_2, \tau_3)$-threshold
      map to a Hilbert space (which without loss of generality can be
      chosen to be $\ell_2$, since $X'$ is finite).  Now, using
      Claim~\ref{sdp_duality}, we obtain finitely-supported
      probability measures $\mu_1$ and $\mu_2$, which
      satisfy~\eqref{eqn:poincareTaus}. This concludes the proof of
      Lemma~\ref{lem:Poincare2Thresh}, since its statement only
      considers $f$ such that the image of $f$ is the unit ball of
      $\ell_2$, and, thus, $\sup_{x\in X} \|f(x)\|^2 \leq 1$. Note
      that the measures $\mu_1,\mu_2$ depend on the value of $\tau_3$
      (and, as a result, on $\delta$).  
    \end{proof}

\subsection{Threshold maps imply uniform embeddings}
\label{sec:Threshold2Uniform}

We now prove that threshold embeddings imply uniform embeddings,
formalized as follows.

    \begin{theorem}\label{thm:Threshold2Uniform}
        \label{symmetric_emb}
        Suppose that $X$ is a finite-dimensional
        normed space such that there exists a $(1, D, \tau_1, \tau_2, +\infty)$-threshold
        map to a Hilbert space for some $D > 1$ and for some $0 < \tau_1 < \tau_2$ with $\tau_2 > 8 \tau_1$.
        Then there exists a map $h$ of $X$ into a Hilbert space
        such that for every $x_1, x_2 \in X$,
        \begin{equation}
            \label{moduli_bound}
            (\tau_2^{1/2} - (8 \tau_1)^{1/2}) \cdot \min \set{1, \frac{\|x_1 - x_2\|}{2D + 4}}
            \leq
            \|h(x_1) - h(x_2)\| \leq (2 \tau_1 \|x_1 - x_2\|)^{1/2}.
        \end{equation}
        In particular, $h$ is a uniform embedding of $X$ into a Hilbert space
        with moduli that depend only on $\tau_1, \tau_2$ and $D$.
    \end{theorem}

Let us point out that in~\cite{JR06,R06}, Johnson and Randrianarivony
prove that for a Banach space coarse embeddability into a Hilbert
space is equivalent to uniform embeddability.
Our definition of a threshold map is weaker than that of a coarse
embedding (for the latter see~\cite{JR06} say), but we show that we can adapt the proof of \cite{JR06,R06}
to our setting as well (at least whenever the gap between $\tau_1$ and
$\tau_2$ is large enough).  Since we only need one direction of the
equivalence, we present a part of the argument from~\cite{JR06} with
one (seemingly new) addition: Claim~\ref{smooth_lb}.  The resulting
proof is arguably simpler than the combination of~\cite{JR06}
and~\cite{R06}, and yields a clean quantitative
bound~\eqref{moduli_bound}.

\myparagraph{Intuition.} Let us provide some very high-level intuition of the proof of Theorem~\ref{thm:Threshold2Uniform}.
We start with a threshold map $f$ from $X$ to a Hilbert space. First, we show that $f$ is Lipschitz on pairs of points
that are sufficiently far. In particular, $f$, restricted on a sufficiently crude net $N$ of $X$,
is Lipschitz. This allows us to use a certain extension theorem to extend the restriction of $f$ on $N$ to a Lipschitz function
on the whole $X$, while preserving the property that $f$ does not contract too much distances that are sufficiently large.
Then, we get a required uniform embedding by performing a certain symmetrization step.

The actual proof is different in a number of details; in particular, instead of being Lipschitz the actual property we will be
trying to preserve is different.

\myparagraph{Useful facts.}
To prove Theorem~\ref{thm:Threshold2Uniform}, we need the following
three results.

    \begin{lemma}[\cite{Schoenberg37}]
        \label{snowflake}
        For a set $S$ and a map $f$ from $S$ to a Hilbert space, there exists a map $g$ from $S$ to a Hilbert space
        such that $\|g(x_1) - g(x_2)\| = \|f(x_1) - f(x_2)\|^{1/2}$ for every $x_1, x_2 \in S$.
    \end{lemma}
    \begin{lemma}[essentially Lemma~3.5 from \cite{AMM85}, see also Lemma~\ref{symmetrization_dot}
from the present paper]
        \label{symmetrization}
        Suppose that $f$ is a map from an abelian group $G$ to a Hilbert space such that
        for every $g \in G$ we have $\sup_{g_1 - g_2 = g} \|f(g_1) - f(g_2)\| < +\infty$.
        Then, there exists a map $f'$ from $G$ to a Hilbert space such that $\|f'(g_1) - f'(g_2)\|$
        depends only on $g_1 - g_2$ and for every $g_1, g_2 \in G$ we have
        \begin{equation}
          \label{bounds_emb}
            \inf_{g_1' - g_2' = g_1 - g_2} \|f(g_1') - f(g_2')\| \leq \|f'(g_1) - f'(g_2)\|
            \leq \sup_{g_1' - g_2' = g_1 - g_2} \|f(g_1') - f(g_2')\|.
        \end{equation}
    \end{lemma}
    \begin{proof}
      This lemma is similar to Lemma~\ref{symmetrization_dot} with one twist: in the statement,
      we now have distances instead of dot products.
The proof of Lemma~\ref{symmetrization_dot} relies on the characterization
      from Fact~\ref{char_dot}.
      If instead we use Fact~\ref{char_dist}, we can reuse the proof of Lemma~3.5 from~\cite{AMM85}
      verbatim to prove
      the present lemma.

      Let us sketch here the symmetrization procedure. Let $B(G)$ be the vector space
      of bounded functions $h \colon G \to \Rbb$. Then, one can show that there exists
      a \emph{finitely additive invariant mean} $M \colon B(G) \to \Rbb$:
      a \emph{linear} functional such that
      \begin{itemize}
        \item for every $h \in B(G)$ such that $h \geq 0$ one has $Mh \geq 0$;
        \item for every $h \in B(G)$ and $g \in G$ one has
$M h = M (x \mapsto h(x + g))$;
        \item $M (x \mapsto 1) = 1$.
      \end{itemize}
      The existence of such $M$ is non-trivial and requires the axiom of choice
      (see Theorem~17.5 from~\cite{HR94}).

      Let us now consider a map $f$ from the statement of the lemma
      and consider the kernel $K(g_1, g_2) = \|f(g_1) - f(g_2)\|^2$.
      Let us define a new function $K'(g_1, g_2)$ as follows:
$$
K'(g_1, g_2) = M (x \mapsto K(x + g_1 - g_2, x)).
$$
Now we need to check that:
\begin{itemize}
\item $K'$ is a kernel (that is, it is non-negative and symmetric)
and $K'(g, g) = 0$ for every $g \in G$;
\item $K'$ is negative-definite (see Definition~\ref{def_pos_neg_kernel}), assuming
that $K$ is negative-definite (which is true by Fact~\ref{char_dist});
\item for every $g_1, g_2 \in G$ one has
$$
            \inf_{g_1' - g_2' = g_1 - g_2} \|f(g_1') - f(g_2')\|^2 \leq K'(g_1, g_2)
            \leq \sup_{g_1' - g_2' = g_1 - g_2} \|f(g_1') - f(g_2')\|^2
$$
assuming~(\ref{bounds_emb}).
\end{itemize}
This can be done exactly the same way as in the proof of Lemma~3.5 from~\cite{AMM85}.
Finally, we observe that $K'(g_1, g_2)$ depends only on $g_1 - g_2$ and via Fact~\ref{char_dist}
gives a map $f'$ from $G$ to a Hilbert space with the required properties.
    \end{proof}
    \begin{definition}
        We say that a map $f \colon X \to Y$ between metric spaces is
        $1/2$-H\"{o}lder with constant $C$, if for every $x_1, x_2 \in X$ one has
        $d_Y(f(x_1), f(x_2)) \leq C \cdot d_X(x_1, x_2)^{1/2}$.
    \end{definition}
    \begin{theorem}[\cite{M70}, see also Theorem 19.1 in~\cite{WW75}).]
        \label{extension}
        Let $(X, d_X)$ be a metric space and let $H$ be a Hilbert space.
        Suppose that $f \colon S \to H$, where $S \subset X$, is a $1/2$-H\"{o}lder map
        with a constant $C > 0$. Then there exists a map $g \colon X \to H$ that coincides
        with $f$ on $S$ and is $1/2$-H\"{o}lder with the constant $C$.
    \end{theorem}

We are now ready to prove Theorem~\ref{thm:Threshold2Uniform}.

\begin{proof}[Proof of Theorem~\ref{thm:Threshold2Uniform}]

    We prove the theorem via the following sequence of claims.
    Suppose that $X$ is a finite-dimensional
    normed space. Let $f$ be a $(1, D, \tau_1, \tau_2, +\infty)$-threshold map
    to a Hilbert space.

    The first claim is well-known and is a variant of Proposition~1.11
    from~\cite{BL00}.

    \begin{claim}
        \label{getting_lip}
        For every $x_1, x_2 \in X$ we have
            $\|f(x_1) - f(x_2)\| \leq \max\set{1, 2 \cdot \|x_1 - x_2\|} \cdot \tau_1$.
    \end{claim}
    \begin{proof}
        If $\|x_1 - x_2\| \leq 1$, then $\|f(x_1) - f(x_2)\| \leq
        \tau_1$, and we are done.  Otherwise, let us take $y_0, y_1,
        \ldots, y_l \in X$ such that $y_0 = x_1$, $y_l = x_2$, $\|y_i
        - y_{i+1}\| \leq 1$ for every $i$, and $l = \lceil \|x_1 -
        x_2\|\rceil$.  In particular, we can take $y_i=x_1+i\cdot
        \tfrac{x_1-x_2}{\|x_1-x_2\|}$ for $i=0,1,\ldots l-1$, and $y_l=x_2$. We have
        $$
            \|f(x_1) - f(x_2)\| \leq \sum_{i=0}^{l-1} \|f(y_i) - f(y_{i+1})\| \leq
            l \tau_1 = \lceil \|x_1 - x_2\| \rceil \cdot \tau_1 \leq 2 \|x_1 - x_2\| \cdot \tau_1, 
        $$
        where
        the first step is by the triangle inequality, the second step follows from $\|y_i - y_{i+1}\| \leq 1$,
        and the last step follows from $\|x_1 - x_2\| \geq 1$.
    \end{proof}

    The proof of the next claim essentially appears in~\cite{JR06}.
    \begin{claim}
        \label{main_claim}
        There exists a map $g$ from $X$ to a Hilbert space
        such that for every $x_1, x_2 \in X$,
        \begin{itemize}
            \item $\|g(x_1) - g(x_2)\| \leq (2\tau_1 \cdot \|x_1 - x_2\|)^{1/2}$;
            \item if $\|x_1 - x_2\| \geq D + 2$, then $\|g(x_1) - g(x_2)\| \geq \tau_2^{1/2} - (8\tau_1)^{1/2}$;
        \end{itemize}
    \end{claim}
    \begin{proof}
        From Claim~\ref{getting_lip} and Lemma~\ref{snowflake} we can get a map
        $g'$ from $X$ to a Hilbert space such that for every $x_1, x_2 \in X$
        \begin{itemize}
            \item $\|g'(x_1) - g'(x_2)\| \leq \max\set{1, (2\|x_1 - x_2\|)^{1/2}}
            \cdot \tau_1^{1/2}$;
            \item if $\|x_1 - x_2\| \geq D$, then $\|g'(x_1) - g'(x_2)\| \geq
            \tau_2^{1/2}$.
        \end{itemize}

        Let $N \subset X$ be a $1$-net of $X$ such that all the pairwise distances between points in $N$
        are more than $1$.
        The map $g'$ is $1/2$-H\"{o}lder on $N$ with a constant
        $(2\tau_1)^{1/2}$,
        so we can apply Theorem~\ref{extension} and get a map $g$ that coincides
        with $g'$ on $N$ and is $1/2$-H\"{o}lder on the whole $X$ with a constant
        $(2\tau_1)^{1/2}$.
        That is, for every $x_1, x_2 \in X$ we have
        \begin{itemize}
            \item $\|g(x_1) - g(x_2)\| \leq (2\tau_1 \cdot \|x_1 - x_2\|)^{1/2}$;
            \item if $x_1 \in N$, $x_2 \in N$ and $\|x_1 - x_2\| \geq D$, then
            $\|g(x_1) - g(x_2)\| \geq \tau_2^{1/2}$.
        \end{itemize}

        To conclude that $g$ is as required, let us lower bound
        $\|g(x_1) - g(x_2)\|$ whenever $\|x_1 - x_2\| \geq D + 2$.
        Suppose that $x_1, x_2 \in X$ are such that $\|x_1 - x_2\| \geq D + 2$.
        Let $u_1\in N$ be the closest net point to $x_1$ and, similarly,
        let $u_2 \in N$ be the closest net point to $x_2$.
        Observe that
        $$
            \|u_1 - u_2\| \geq \|x_1 - x_2\| - \|x_1 - u_1\| - \|x_2 - u_2\| \geq (D + 2) - 1 - 1 = D.
        $$
        We have
        $$
            \|g(x_1) - g(x_2)\| \geq \|g(u_1) - g(u_2)\| - \|g(u_1) - g(x_1)\| - \|g(u_2) - g(x_2)\|
            \geq \tau_2^{1/2} - 2 (2 \tau_1)^{1/2},
        $$
        as required, where the second step follows from the inequality $\|g(u_1) - g(u_2)\| \geq \tau_2^{1/2}$,
        which is true, since $u_1, u_2 \in N$, and that $g$ is 1/2-H\"{o}lder with a constant $(2 \tau_1)^{1/2}$.
    \end{proof}

    The following claim completes the proof of Theorem~\ref{thm:Threshold2Uniform}.
    \begin{claim}
        \label{smooth_lb}
        There exists a map $h$ from $X$ to a Hilbert space such that
        for every $x_1, x_2 \in X$:
        \begin{itemize}
            \item $\|h(x_1) - h(x_2)\| \leq (2 \tau_1 \cdot \|x_1 - x_2\|)^{1/2}$;
            \item $\|h(x_1) - h(x_2)\| \geq (\tau_2^{1/2} - (8 \tau_1)^{1/2}) \cdot
            \min \set{1, \|x_1 - x_2\| / (2D + 4)}$.
        \end{itemize}
    \end{claim}
    \begin{proof}

        We take the map $g$ from Claim~\ref{main_claim} and apply Lemma~\ref{symmetrization}
        to it. Let us call the resulting map $h$.
        The first desired condition for $h$ follows from a similar condition for $g$ and
        Lemma~\ref{symmetrization}. Let us prove the second one.

        If $x_1 = x_2$, then there is nothing to prove.
        If $\|x_1 - x_2\| \geq D + 2$, then by Claim~\ref{main_claim} and Lemma~\ref{symmetrization},
        $\|h(x_1) - h(x_2)\| \geq
        \tau_2^{1/2} - (8\tau_1)^{1/2}$, and we are done.
        Otherwise,
        let us consider points $y_0, y_1, \ldots, y_l \in X$
        such that $y_0 = 0$, $y_i - y_{i-1} = x_1 - x_2$ for
        every $i$, and $l = \left\lceil \frac{D + 2}{\|x_1 - x_2\|}\right\rceil$.
        Since $\|y_l - y_0\| = \|l(x_1 - x_2)\| = l\|x_1 - x_2\| \geq
        D + 2$, we have
        \begin{multline*}
            \tau_2^{1/2} - (8\tau_1)^{1/2} \leq \|h(y_l) - h(y_0)\|
            \leq \sum_{i=1}^{l} \|h(y_i) - h(y_{i-1})\| \\= l \cdot \|h(x_1) - h(x_2)\|
            \leq \frac{2D + 4}{\|x_1 - x_2\|} \cdot \|h(x_1) - h(x_2)\|,
        \end{multline*}
        where the equality follows from the conclusion of Lemma~\ref{symmetrization}.
    \end{proof}

    Finally, observe that Theorem~\ref{thm:Threshold2Uniform} is merely
    a reformulation of Claim~\ref{smooth_lb}.
\end{proof}

\subsection{Putting it all together}
\label{sec:together}

We now show that Theorem~\ref{main_quant} follows by applying
Lemma~\ref{lem:fold}, 
Theorem~\ref{thm:AJP},
Lemma~\ref{lem:Poincare2Thresh}, and
Theorem~\ref{thm:Threshold2Uniform}, in this order,
with an appropriate choice of parameters.

\begin{proof}[Proof of Theorem~\ref{main_quant}]
Suppose $\DTEP(X, D)$ admits a protocol of size $s$. 
By setting $k = Cs$ in Lemma~\ref{lem:fold} ($C$ is a large absolute constant, to be chosen later),
we conclude that $\DTEP(\ell_\infty^{Cs}(X), CsD)$ admits a protocol 
of size $s'=O(s)$.

Now choosing $C$ large enough and applying Theorem~\ref{thm:AJP} (in contrapositive),
we conclude that $X$ has no Poincar\'{e} inequalities for
distance scales $1$ and $CsD$, with $\alpha = 0.01$ and $\beta = 0.001$.

Applying Lemma~\ref{lem:Poincare2Thresh} (in contrapositive)
we conclude that 
$X$ allows a $(1,CsD,1,10,+\infty)$-threshold map to a Hilbert space.

Using Theorem~\ref{thm:Threshold2Uniform} it follows that
there is a map $h$ from $X$ to a Hilbert space, 
such that for all $x_1, x_2 \in X$,
$$
  \min \set{1, \frac{\|x_1 - x_2\|}{s \cdot D}}
  \leq \norm{h(x_1) - h(x_2)}
  \leq K \cdot \norm{x_1 - x_2}^{1/2},
$$
where $K>1$ is an absolute constant, 
and this proves the theorem.
\end{proof}

\myparagraph{Remark:} Instead of applying Lemma~\ref{lem:fold} and Theorem~\ref{thm:AJP},
we could have attempted to apply the reduction from~\cite{AK07} to get a threshold map
from $X$ to a Hilbert space directly. 
That approach is much simpler technically, but has two fatal drawbacks.
First, we end up with a threshold map with a gap between $\tau_1$ and $\tau_2$ being arbitrarily close to $1$,
and thus, we are unable to invoke Theorem~\ref{thm:Threshold2Uniform}, which requires the gap to be more than $8$.
Second, the parameters of the resulting threshold map are \emph{exponential} in the number of bits
in the communication protocol, which is bad for the quantitative bounds from Section~\ref{quant_sect}.

\section{Quantitative bounds}
\label{quant_sect}

In this section we prove the quantitative version of our results,
namely Theorem~\ref{thm:l_1eps} and Theorem~\ref{thm:l_1log}, for
which we will reuse Theorem~\ref{main_quant}. In particular, we
prove the following theorem.

\begin{theorem}
\label{thm:unifL1quant}
For a finite-dimensional normed space $X$ and $\Delta>1$, assume we have a map $f \colon X \to H$
to a Hilbert space $H$, such that, for an absolute constant $K > 0$
and for every $x_1, x_2 \in X$:
\begin{itemize} \compactify
\item $\|f(x_1) - f(x_2)\|_H \leq K \cdot \|x_1 - x_2\|_X^{1/2}$; and
\item if $\|x_1 - x_2\|_X \geq \Delta$, then $\|f(x_1) - f(x_2)\|_H \geq 1$.
\end{itemize}

Then, for any $\eps\in (0,1/3)$, the space $X$ linearly embeds into
$\ell_{1-\eps}$ with distortion $O(\Delta/\epsilon)$.
\end{theorem}

Note that Theorem~\ref{thm:l_1eps} now follows from applying
Theorem~\ref{main_quant} together with Theorem~\ref{thm:unifL1quant}
for $\Delta=sD$. 
We can further prove Theorem~\ref{thm:l_1log}
by using the following result of Zvavitch from~\cite{Z00}.

\begin{lemma}[\cite{Z00}]
\label{lem:Z00}
  Every $d$-dimensional subspace of $L_{1 - \eps}$ embeds linearly into $\ell_{1 - \eps}^{d \cdot \poly(\log d)}$
  with distortion $O(1)$.
\end{lemma}

Indeed, applying Lemma~\ref{lem:Z00} together with Theorem~\ref{thm:unifL1quant}, 
we get that for every $0 < \eps < 1/3$ the space $X$ linearly embeds
into $\ell_{1 - \eps}^{\poly(\dim X)}$ with distortion $O(\Delta / \eps)$.
Thus, $X$ is embeddable into $\ell_1$ with distortion
$$
O\bigl(\Delta \cdot (\dim X)^{O(\eps)}/ \eps\bigr).
$$
Setting $\eps = \Theta(1 / \log (\dim X))$, we obtain Theorem~\ref{thm:l_1log}.



It remains to prove Theorem~\ref{thm:unifL1quant}. Its proof proceeds by
adjusting the arguments from~\cite{AMM85} and~\cite{Nikishin72}.

\begin{proof}[Proof of Theorem~\ref{thm:unifL1quant}]
Fix $X$, $\Delta>0$, and the corresponding map $f \colon X \to H$.
We first prove the following lemma.

\begin{lemma}
  \label{summary_fourier}
  There exists a probability measure $\mu$ on $\Rbb^{\dim X}$ symmetric around the origin
  such that its (real-valued) characteristic function $\varphi:X\to \R$
  has the following properties for every $x \in X$:
  \begin{itemize} \compactify
  \item $\varphi(x) \geq e^{-\widetilde{K} \cdot \|x\|_X}$; and
  \item if $\|x\|_X \geq \Delta$, then $\varphi(x) \leq 1 / e$.
  \end{itemize}
  Here $\widetilde{K} > 0$ is an absolute constant.
\end{lemma}

\begin{proof}
It is known from~\cite{Sch38} that for a Hilbert space $H$ the function $g \colon h \mapsto e^{-\|h\|_H^2}$ is positive-definite.
Thus, there exists a function $\widetilde{g} \colon H \to \widetilde{H}$ to a Hilbert space $\widetilde{H}$
such that for every $h_1, h_2 \in H$
one has $\bigl\langle \widetilde{g}(h_1), \widetilde{g}(h_2)\bigr\rangle_{\widetilde{H}} = e^{-\|h_1 - h_2\|_H^2}$.
Setting $\widetilde{f} = \widetilde{g} \circ f$, we get a function $\widetilde{f} \colon X \to \widetilde{H}$ to a Hilbert space
such that for an absolute constant $\widetilde{K} > 0$ for every $x_1,
x_2 \in X$, we have:
\begin{itemize} \compactify
\item $\Bigl\|\widetilde{f}(x_1)\Bigr\|_{\widetilde{H}} = 1$;
\item $\Bigl\langle \widetilde{f}(x_1), \widetilde{f}(x_2)\Bigr\rangle_{\widetilde{H}} \geq e^{-\widetilde{K} \cdot \|x_1 - x_2\|_X}$; and
\item if $\|x_1 - x_2\|_X \geq \Delta$, then $\Bigl\langle \widetilde{f}(x_1), \widetilde{f}(x_2)\Bigr\rangle_{\widetilde{H}} \leq 1/e$.
\end{itemize}

Applying Lemma~\ref{symmetrization_dot} and 
Lemma~\ref{char_dot}, we obtain a
positive-definite function $\varphi \colon X \to \Rbb$ such that:
\begin{itemize} \compactify
\item $\varphi(0) = 1$;
\item for every $x \in X$ one has $\varphi(x) \geq e^{-\widetilde{K} \cdot \|x\|_X}$; and
\item if $\|x\|_X \geq \Delta$, then $\varphi(x) \leq 1/e$.
\end{itemize}

We can now use Bochner's theorem, which is the following characterization of continuous positive-definite functions, via the Fourier transform.

\begin{theorem}[Bochner's theorem, see~\cite{Feller}]
  \label{bochner}
  If a function $f \colon \Rbb^d \to \Rbb$ is positive-definite, continuous at zero, and $f(0) = 1$, then there exists a probability measure
  $\mu$ on $\Rbb^d$ such that $f$ is the $\mu$'s characteristic function. That is, for every $x \in \Rbb^d$,
  $$
  f(x) = \int_{\Rbb^d} e^{i \langle x, v\rangle}\, \mu(dv).
  $$
\end{theorem}

In particular, note that we have that $\varphi(0) = 1$, $\varphi$ is
positive-definite and is continuous at zero. Hence,
by Bochner's theorem, we get a probability measure $\mu$
over $\Rbb^{\dim X}$ whose characteristic function equals to $\varphi$. That is, for every $x \in X$ we get
$$
\varphi(x) = \int_{\Rbb^{\dim X}} e^{i \langle x, v\rangle}\,\mu(dv),
$$
where $\langle \cdot, \cdot\rangle$ is the standard dot product in $\Rbb^{\dim X}$.
Clearly, $\mu$ is symmetric around the origin, since $\varphi$ is
real-valued. This completes the proof of Lemma~\ref{summary_fourier}.
\end{proof}

Our next goal is to show that $\mu$ gives rise to a one-measurement \emph{linear} sketch for $X$ with approximation $O(\Delta)$
and a certain additional property that will be useful to us.
The following lemma contains two standard facts about \emph{one-dimensional} characteristic functions
(see, e.g.,~\cite{panchenko}).
We include the proof for completeness.

\begin{lemma}
  \label{char_f}
  Let $\nu$ be a symmetric probability measure over the real line, and let
  $$
  \psi(t) = \int_{\Rbb} e^{ivt} \, \nu(dv)
  $$
  be its characteristic function (which is real-valued due to the symmetry of $\nu$).
  Then,
  \begin{itemize}
  \item if for some $R > 0$ and $0 < \eps < 1$ we have $|\psi(R)| \leq 1 - \eps$, then
    \begin{equation}
      \label{char_f_1}
    \nu\Bigl(\bigl\{v \in \Rbb \colon |v| \geq \Omega_{\eps}(1 / R)\bigr\}\Bigr) \geq \Omega_{\eps}(1);
    \end{equation}
  \item for every $\delta > 0$ one has
    \begin{equation}
      \label{char_f_2}
    \nu\Bigl(\bigl\{v \in \Rbb \colon |v| \geq 1 / \delta\bigr\}\Bigr) \leq O(1 / \delta) \cdot \int_{-\delta}^{\delta} \bigl(1 - \psi(t)\bigr) \, dt.
    \end{equation}
  \end{itemize}
\end{lemma}
\begin{proof}
  Let us start with proving the first claim.
  We have for every $\alpha > 0$
  \begin{multline*}
  1 - \eps \geq |\psi(R)| \geq \int_\Rbb \cos(vR) \,\nu(dv) \geq \cos \alpha \cdot \nu\Bigl(\{v \in \Rbb \colon |vR| \leq \alpha\}\Bigr) - \nu \Bigl(\{v \in \Rbb\colon |vR| > \alpha\}\Bigr) \\ = (1 + \cos \alpha) \cdot \nu\Bigl(\{v \in \Rbb \colon |vR| \leq \alpha\}\Bigr) - 1,
  \end{multline*}
  where the second step uses the fact that $\psi$ is real-valued. Thus, we have
  $$
  \nu\Bigl(\{v \in \Rbb \colon |vR| \leq \alpha\}\Bigr) \leq \frac{2 - \eps}{1 + \cos \alpha}.
  $$
  Setting $\alpha = \Theta\bigl(\sqrt{\eps}\bigr)$, we get the desired bound.

  Now let us prove the second claim. We have, for every $\delta > 0$,
  \begin{multline*}
  \int_{-\delta}^{\delta} \bigl(1 - \psi(t)\bigr) \, dt = \int_{-\delta}^{\delta} \int_{\Rbb} \bigl(1 - e^{ivt}\bigr) \, \nu(dv) \, dt
  = 2 \delta \cdot \int_{\Rbb} \left(1 - \frac{\sin(\delta v)}{\delta v}\right) \, \nu(dv) \\ \geq
  2 (1 - \sin 1) \cdot \delta \cdot \nu\Bigl(\{v \in \Rbb \colon |\delta v| \geq 1\}\Bigr),
  \end{multline*}
  where we use that $(1 - \sin y / y) > (1 - \sin 1)$ for every $y$ such that $|y| > 1$.
\end{proof}

Now we will show that the probability measure $\mu$ from Lemma~\ref{summary_fourier} gives a good linear sketch for $X$.
To see this we use the conditioning on the characteristic function of $\mu$, namely, we exploit them using the above Lemma~\ref{char_f}.
In order to do this, we look at the one-dimensional projections of $\mu$ as follows.
Let $x \in X$ be a fixed vector.
For a measurable subset $A \subseteq \Rbb$ we define
$$
\nu(A) = \mu\Bigl(\bigl\{v \in \Rbb^{\dim X} \colon \langle x, v \rangle \in A\bigr\}\Bigr).
$$
It is immediate to check that the characteristic function $\psi$ of $\nu$ is as follows: $\psi(t) = \varphi(t \cdot x)$
(recall that $\varphi$ is the characteristic function of $\mu$).
Next we apply Lemma~\ref{char_f} to $\psi$ and use the properties of $\varphi$ from the conclusion of Lemma~\ref{summary_fourier}.
Namely, we get for every $x \in X$:
\begin{equation}
  \label{sketch_1}
\mu\Bigl(\bigl\{v \in \Rbb^{\dim X} \colon |\langle x, v \rangle| \geq \Omega(\|x\|_X / \Delta)\bigr\}\Bigr) = \Omega(1);
\end{equation}
and for every $t > 0$,
\begin{equation}
  \label{sketch_2}
  \mu\Bigl(\bigl\{v \in \Rbb^{\dim X} \colon |\langle x, v \rangle| \geq t \cdot \|x\|_X\bigr\}\Bigr) \leq O(1 / t).
  \end{equation}

Indeed,~(\ref{sketch_1}) follows from the bound $\varphi(x) \leq 1 /
e$ whenever $\|x\|_X \geq \Delta$
and~(\ref{char_f_1}). The inequality~\eqref{sketch_2} follows from the estimate
$\varphi(x) \geq e^{-\widetilde{K} \cdot \|x\|_X}$ and \eqref{char_f_2}
(for $1/\delta=t\|x\|_X$) that together give
\begin{align*}
\mu\Bigl(\bigl\{v \in \Rbb^{\dim X} \colon |\langle x, v \rangle| \geq
t \cdot \|x\|_X\bigr\}\Bigr)
&=
\nu(\{r\in \R \colon |r|\ge t\|x\|_X\})
\\
&\le
O(t\|x\|_X)\int_{-1/t\|x\|_X}^{1/t\|x\|_X}(1-e^{-\tilde K
  s\|x\|_X}) \,ds,
\end{align*}
as well as from the inequality 
$$
t\int_{-1/t}^{1/t} \bigl(1 - e^{-Cs}\bigr) \,ds \leq 
t\int_{-1/t}^{1/t} \bigl(1 - (1-Cs)\bigr) \,ds =
C/t.
$$
  
Hence, $\mu$ gives rise to a one-measurement linear sketch of $X$ with approximation $O(\Delta)$, whose ``upper tail''
is not too heavy.

Finally, we are ready to describe a desired linear embedding of $X$ into $L_{1 - \eps}$; we map $X$ into $L_{1 - \eps}(\mu)$
as follows: $x \mapsto (v \mapsto \langle x, v\rangle)$. The following Lemma states that the distortion of this embedding
is $O(\Delta / \eps)$, as required. 

\begin{lemma}
  For $0 < \eps < 1/3$ and every $x \in X$,
  $$
  \Omega\bigl(\|x\|_X / \Delta\bigr) \leq \|v \mapsto \langle x, v\rangle\|_{L_{1 - \eps}(\mu)} \leq O\bigl(\|x\|_X / \eps\bigr).
  $$
\end{lemma}
\begin{proof}
  The lower bound is straightforward:
  $$
  \|v \mapsto \langle x, v\rangle\|_{L_{1 - \eps}(\mu)}^{1 - \eps} = \int_{\Rbb^{\dim X}} \bigl|\langle x, v \rangle\bigr|^{1 - \eps} \, \mu(dv)
  \geq \Omega(1) \cdot \Omega(\|x\|_X / \Delta)^{1 - \eps},
  $$
  where the last step follows from~(\ref{sketch_1}).

  For the upper bound, we have for every $\alpha > 0$,
  \begin{multline*}
  \|v \mapsto \langle x, v\rangle\|_{L_{1 - \eps}(\mu)}^{1 - \eps} = \int_{\Rbb^{\dim X}} \bigl|\langle x, v \rangle\bigr|^{1 - \eps} \, \mu(dv)
  = \int_{0}^{\infty} \mu\Bigl(\{v \in \Rbb^{\dim X} \colon |\langle x, v\rangle|^{1 - \eps} \geq s\}\Bigr) \, ds
  \\ \leq \alpha + O(1) \cdot \|x\|_X \cdot \int_{\alpha}^{\infty} s^{-\frac{1}{1 - \eps}}\, ds
  \leq \alpha + O(1) \cdot \frac{\|x\|_X}{\eps} \cdot \alpha^{-\frac{\eps}{1 - \eps}},
  \end{multline*}
  where the third step follows from~(\ref{sketch_2}).
  Choosing $\alpha = \|x\|_X^{1 - \eps} / \eps^{1 - \eps}$, we get
  $$
  \|v \mapsto \langle x, v\rangle\|_{L_{1 - \eps}(\mu)} \leq O(1) \cdot 2^{\frac{1}{1 - \eps}} \cdot \frac{\|x\|_X}{\eps}.
  $$
\end{proof}

This concludes the proof of Theorem~\ref{thm:unifL1quant}.
\end{proof}


\section{Embedding into $\ell_1$ via sum-products}
\label{sec:L1Embedding}

Finally, we prove Theorem~\ref{thm:l_1}: good sketches for norms closed under the sum-product imply embeddings into $\ell_1$
with constant distortion. First we invoke Theorem~\ref{main_quant} and get a sequence of good uniform embeddings into a Hilbert space,
whose moduli depend only on the sketch size and the approximation. Then, we use the main result of this section: Lemma~\ref{lem:L1Embedding}.
Before stating the lemma, let us
remind a few notions. For a metric space $X$, recall
that the metric space $\ell_1^k(X) =
\bigoplus_{\ell_1}^k X_n$ 
is the direct sum of $k$ copies of $X$, with the associated
distance defined as a sum-product ($\ell_1$-product) over the $k$
copies.  We define $\ell_1(X)$ similarly. We also denote $X \oplus_{\ell_1} Y$
the sum-product of $X$ and $Y$.

\begin{lemma} \label{lem:L1Embedding}
    Let $(X_n)_{n=1}^{\infty}$ be a sequence of finite-dimensional normed spaces.
    Suppose that for every $i_1, i_2 \geq 1$ there exists $m = m(i_1, i_2) \geq 1$ 
    such that $X_{i_1} \oplus_{\ell_1} X_{i_2}$
    is isometrically embeddable into $X_m$.
    If every $X_n$ admits a uniform embedding into a Hilbert space
    with moduli independent of $n$,
    then every $X_n$ is linearly embeddable into $\ell_1$ with distortion independent of $n$.
\end{lemma}

Note that Theorem \ref{thm:l_1} just follows from combining Lemma
\ref{lem:L1Embedding} with Theorem \ref{main_quant}.

Before proving Lemma~\ref{lem:L1Embedding}, we state the following two
useful theorems.  The first one (Theorem~\ref{ultrafilters}) follows from the
fact that uniform embeddability into a Hilbert space is determined by
embeddability of finite subsets \cite{BL00}.  The second one
(Theorem~\ref{thm:kalton}) follows by composing results of Aharoni,
Maurey, and Mityagin~\cite{AMM85} and Kalton~\cite{Kalton85}.

\begin{theorem}[Proposition 8.12 from~\cite{BL00}] 
    \label{ultrafilters}
    Let $A_1 \subset A_2 \subset \ldots$ be metric spaces and let
    $A = \bigcup_i A_i$.
    If every $A_n$ is uniformly embeddable into a Hilbert space
    with moduli independent of $n$, 
    then the whole $A$ is uniformly embeddable into a Hilbert space.
\end{theorem}

\begin{theorem}[\cite{AMM85, Kalton85}]
    \label{thm:kalton}
    A Banach space $X$ is linearly embeddable into $L_1$ iff $\ell_1(X)$ is uniformly embeddable into a Hilbert
    space.
\end{theorem}

We are now ready to proceed with the proof of Lemma~\ref{lem:L1Embedding}.

\begin{proof}[Proof of Lemma~\ref{lem:L1Embedding}]
  Let $X = X_1 \oplus_{\ell_1} X_2 \oplus_{\ell_1} \ldots$. More formally,
  $$
  X = \Bigl\{(x_1, x_2, \ldots) : x_i \in X_i, \sum_i \|x_i\| < \infty\Bigr\},
  $$
  where the norm is defined as follows:
  $$
  \bigl\|(x_1, x_2, \ldots)\bigr\| = \sum_i \|x_i\|.
  $$
    We claim that the space $\ell_1(X)$ embeds uniformly into a Hilbert space.
    To see this, consider $U_p = \ell_1^p(X_1 \oplus_{\ell_1} X_2 \oplus_{\ell_1} \ldots \oplus_{\ell_1} X_p)$,
    which can be naturally seen as a subspace of $\ell_1(X)$.
    Then, $U_1 \subset U_2 \subset \ldots \subset U_p \subset \ldots \subset \ell_1(X)$ and
        $\bigcup_p U_p$ is dense in $\ell_1(X)$.
        By the assumption of the lemma, $U_p$ is isometrically embeddable into $X_m$
        for some $m$, thus, $U_p$ is uniformly embeddable into a Hilbert space
        with moduli independent of $p$.
        Now, by Theorem~\ref{ultrafilters},
        $\bigcup_p U_p$ is uniformly embeddable into a Hilbert space.
        Since $\bigcup_p U_p$ is dense in $\ell_1(X)$,
        the same holds also for the whole $\ell_1(X)$, as claimed.

        Finally, since $\ell_1(X)$ embeds uniformly into a Hilbert
        space, we can apply Theorem~\ref{thm:kalton} and conclude that
        $X$ is linearly embeddable into $L_1$.  The lemma follows
        since $X$ contains every $X_i$ as a subspace.
\end{proof}

\section{Acknowledgments}

We are grateful to Assaf Naor for pointing us to \cite{AMM85,Kalton85}, 
as well as for numerous very enlightening discussions throughout this
project.
We also thank Gideon Schechtman for useful discussions and explaining
some of the literature. We thank Piotr Indyk for
fruitful discussions and for encouraging us to work on this project.

\singlespacing

{\small
\bibliographystyle{alphaurlinit}
\bibliography{bibfile}
}

\onehalfspacing
\appendix
\section{EMD Reduction}
\label{apx:emdReduction}

Recall that $\EMD_n$ is a normed space on all signed measures on $[n]^2$
(that sum up to zero). We also take the view that a weighted set in
$[n]^2$ is in fact a measure on $[n]^2$.

\begin{lemma}
\label{lem:emdReduction}
Suppose the EMD metric between non-negative measures (of the same total measure)
admits a sketching algorithm $\sk$ with approximation $D > 1$ and sketch size $s$. 
Then the normed space $\EMD_n$ admits a sketching algorithm $\sk'$ 
with approximation $D$ and sketch size $O(s)$.
\end{lemma}

\begin{proof}
The main idea is that if $x,y$ are signed measures and
we add a sufficiently large term $M>0$ to all of their coordinates,
then the resulting vectors $x'=x+M\cdot\vone$ and $y'=y+M\cdot\vone$ 
are measures (all their coordinates are non-negative) of the same total mass, and 
$\norm{x-y}_{\EMD}$ is equal to the EMD distance between measures $x',y'$.
The trouble is in identifying a large enough $M$. 
We use the values of $x$ and $y$ themselves to agree on $M$. Details follow.

Without loss of generality we can fix the $\DTEP$ threshold to be $r=1$.

We design the sketch $\sk'$ as follows. First choose a hash function
$h:\N\to \{0,1\}^9$ (using public randomness). 
Fix an input $x\in\R^{n^2}$ of total measure zero, i.e., $\sum_i x_i=0$.
Let $m(x)=\min_i x_i$, and let $b(x)$ be the largest multiple of $2$ that is smaller than $m(x)$.
Since $x$ has total measure zero, $b(x) < m(x) \leq 0$.
Now let $b^{(1)}(x)=b(x)$ and $b^{(2)}(x)=b(x)-2$,
and then $x^{(q)}=x-b^{(q)}(x)\cdot\vone$ for $q=1,2$.
Notice that in both cases 
$x^{(q)}> x\ge 0$ (component-wise). 
Now let the sketch $\sk'(x)$ be the concatenation of 
$\sk(x^{(q)}),h(b^{(q)}(x))$ for $q=1,2$.

The distinguisher works as follows, given two sketches 
$\sk'(x) = ( \sk(x^{(q)}),h(b^{(q)}(x)) )_{q=1,2}$
and $\sk'(y) = ( \sk(y^{(q)}),h(b^{(q)}(x)) )_{q=1,2}$.
If there are $q_x,q_y\in\{1,2\}$ 
whose hashes agree $h(b^{(q)}(x)) = h(b^{(q)}(y))$
(breaking ties arbitrarily if there are multiple possible agreements),
then output whatever the EMD metric distinguisher would output
on $\sk(x^{(q_x)}), \sk(y^{(q_y)})$. 
Otherwise output ``far'' 
(i.e., that $\|x-y\|_{\EMD}>D$).

To analyze correctness, consider the case when $\|x-y\|_{\EMD}\le
1$. Without loss of generality, suppose $m(x)\ge m(y)$. Then $m(x)-m(y)\le 1$ (otherwise
$x,y$ are further away in EMD norm than 1). 
Hence either $b(x)=b(y)$ or $b(x)=b(y)+2$. 
Then there exists a corresponding $q\in\set{1,2}$ 
for which the hashes agree $h(b^{(1)}(x))=h(b^{(q)}(y))$.
By properties of the hash function, 
with sufficiently large constant probability 
the hashes match only when the $b$'s match,
in which case the values $q_x,q_y$ used by the distinguisher satisfy 
$b^{(q_x)}(x)=b^{(q_y)}(y)$.
In this case, $\|x-y\|_{\EMD} = d_{\EMD}( x^{(q_x)}, y^{(q_y)}) $,
and the correctness now depends on $\sk$, 
and the distinguisher for the EMD metric.

Otherwise, if $\|x-y\|_{\EMD}>D$, either the $b$-values
coincide for some $q_x,q_y$ and then the above argument applies again, 
or with sufficiently large constant probability the hashes will not agree 
and the distinguisher outputs (correctly) ``far''.

There is a small loss in success probability due to use of the hash function, 
but that can be amplified back by independent repetitions.
\end{proof}

Notice that the above lemma assumes a sketching algorithm for the EMD metric
between any non-negative measures of the same total measure, 
and not only in the case where the total measure is $1$.
The proof can be easily modified so that any non-negative measure being used
always has a fixed total measure (say $1$, by simply scaling the inputs),
which translates to scaling the threshold $r$ of the DTEP problem.
This is acceptable because, under standard definitions,
a metric space is called sketchable if it admits a sketching scheme 
for every threshold $r>0$.


\end{document}